\documentclass{IEEEtran}
\usepackage{amsmath}
\usepackage{latexsym}
\usepackage{enumerate}
\usepackage{amssymb}
\usepackage{algorithm2e}
\ifx\pdfoutput\undefined
\usepackage{graphicx}
\else
\usepackage[pdftex]{graphicx}
\usepackage{flushend}
\usepackage{subfigure}
\usepackage[T1]{fontenc}
\usepackage{epstopdf}
\fi
\usepackage{cite}
\newtheorem{theorem}{Theorem}

\title{Influence Spread in Social Networks: A Study via a Fluid Limit of the Linear Threshold Model}
\begin{document}
\author{Srinivasan~Venkatramanan,~\IEEEmembership{Student Member,~IEEE,}
        Anurag~Kumar,~\IEEEmembership{Fellow,~IEEE,}\protect\\
        Department of Electrical Communication Engineering, Indian Institute of Science, \protect\\ Bangalore - 560012, India.\protect\\
        E-mail: vsrini,anurag@ece.iisc.ernet.in
}

\maketitle
\thispagestyle{empty}
\pagestyle{empty}
\begin{abstract}
Threshold based models have been widely used in characterizing collective behavior on social networks. An individual's threshold indicates the minimum level of ``influence'' that must be exerted, by other members of the population engaged in some activity, before the individual will join the activity. In this work, we begin with a homogeneous version of the Linear Threshold model proposed by Kempe et al. \cite{kempe-etal03max-spread-infl} in the context of viral marketing, and generalize this model to arbitrary threshold distributions. We show that the evolution can be modeled as a discrete time Markov chain, and, by using a certain scaling, we obtain a fluid limit that provides an ordinary differential equation model (o.d.e.). We find that the threshold distribution appears in the  o.d.e.\ via its \emph{hazard rate} function. We demonstrate the accuracy of the o.d.e. approximation and derive explicit expressions for the trajectory of influence under the uniform threshold distribution. Also, for an 
exponentially distributed threshold, we show that the fluid dynamics are equivalent to the well-known SIR model in epidemiology. We also numerically study how other hazard functions (obtained from the Weibull and loglogistic distributions) provide qualitative different characteristics of the influence evolution, compared to traditional epidemic models, even in a homogeneous setting. We finally show how the model can be extended to a setting with multiple communities and conclude with possible future directions. 
\end{abstract}


\section{Introduction}
\label{sec:intro}
Social networks play a fundamental role in the spread of information, ideas and influence among its members.  The study of influence spread as a stochastic process has been of interest to sociologists for several decades \cite{bartholomew67stochastic}. Such diffusion processes have been used to characterize collective behavior \cite{granovetter78threshold-models}, adoption of innovations \cite{everett62diffusion-innovation, valente96social-thresholds}, etc. among a population of users. Similar models have also been developed independently in other domains to study epidemics\cite{bailey75math-epidemiology}, synchronization in biological systems\cite{mirollo90sync}, etc.  

Online social networks such as Facebook, and Twitter, with their widespread adoption, have enabled information spread on a scale heretofore unimaginable. A significant fraction of online traffic comprises user-generated content on platforms such as Wordpress (text), Flickr (images), YouTube (videos), etc. In most such platforms, we see that users can obtain information on the global popularity of an item of content, for instance the number of views for a YouTube video. Global metrics such as viewcount provide a crude signal to the user about the quality of the content. Given the limited attention span and vast quantity of content available on the Internet, the tendency of a particular user to view a video or read an article increases with the number of people who have already viewed/read it. Hence a YouTube video with many views or a news article with many ``Likes'' on Facebook is more likely to be accessed than others. The understanding and prediction of popularity evolution \cite{hong11predict-twitter, 
szabo10predict-popularity} of such content is crucial to the content provider for (i) choosing appropriate content caching strategies for quick delivery (ii) deploying better advertisement mechanisms for increased monetization.

 
\textit{Related Work:} Threshold models are well established for modeling the evolution of popularity in human populations. Everett \cite{everett62diffusion-innovation} explored the adoption of innovations, employing examples from  rural sociology, and noted the diversity in people's propensity to adopt an innovation. He categorized them into various adopter groups (see Figure~\ref{fig:everett-bell-curve}) in what is now known as the Everett's bell curve\footnote{Everett's use of the term  ``bell curve'', is not be taken to imply that the thresholds are normally distributed.}. It is used to represent the process of adoption of a new product over time.  Under the threshold interpretation, users in the leftmost class (innovators) can be interpreted as having the least threshold to adopt an innovation, and the ones in the rightmost class (laggards) as having the highest threshold (and hence least susceptible). 

\begin{figure}[t]
\centerline{\includegraphics[scale=0.3]{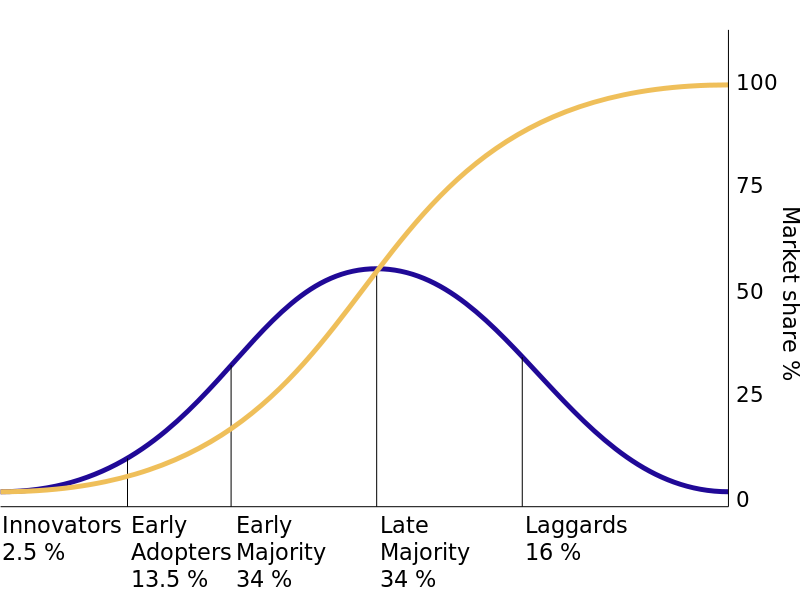}}
\caption{Everett's ``bell curve''. It is also known as the Technology Adoption Lifecycle since it represents the adoption of a new innovation over time. The bell shaped curve represents the incremental change in adoption over time, whereas the ``s'' shaped curve shows the cumulative adoption. Under the threshold interpretation, the early adopters (innovators) can be thought of as having the least threshold for adoption, and so on.(Source \cite{everett62diffusion-innovation})}
\label{fig:everett-bell-curve}
\end{figure}

Granovetter, in his seminal work \cite{granovetter78threshold-models} on collective behavior, aimed to use threshold distributions to model the spread of binary decisions among a group of rational agents, for instance during riots, voting, etc. Using his model, he calculated the equilibrium, i.e., steady state split of the population (between the binary decisions) given the threshold distribution, and also considered the stability of such equilibria. Valente \cite{valente96social-thresholds} further refined this approach to threshold phenomena based on personal networks (local neighbourhood of an individual) as against whole social systems, and empirically studied datasets on the adoption of medical and rural innovation. 

Domingos and Richardson \cite{domingos01mining-network-value} studied influence spread in the context of viral marketing, and they posed the algorithmic question of maximizing the spread of influence, given the underlying social influence network. Kempe et al.\ continued the algorithmic approach in \cite{kempe-etal03max-spread-infl}, where they studied the influence maximization problem under two different activation models (Linear Threshold model and the Independent Cascade model). They proved the submodularity of the influence function (i.e., the set valued function that maps the initial ``seed'' set of adopters to the final set), and provided greedy approximation algorithms to maximize influence spread under the linear threshold model. Recently, in the context of user-generated content on the Internet, game theoretic analysis has shown that threshold based policies could emerge as equilibria when user's seek to maximize their utility (based on the perceived quality of the content) \cite{altman-
etal13emergence-equilibria}.

\textit{Missing Link:}
Our work in this paper is inspired by the earlier efforts to employ a threshold model for the propensity of a user to be influenced by others in the population~ \cite{kempe-etal03max-spread-infl} \cite{granovetter78threshold-models}.  We will now discuss certain modeling details in these two important models of threshold based spread of influence, to motivate our work. Granovetter \cite{granovetter78threshold-models}, rooted in sociology, considered \emph{general influence threshold distributions}, and characterized the spread of influence by a simple difference equation. For instance, if the threshold is distributed with c.d.f. $F(\cdot)$, letting $r(t)$ denote the number of influenced individuals at time $t$, the evolution is described by the difference equation 

\begin{equation}
r(t+1) = F(r(t))
\label{eqn:gran-fixedpoint}
\end{equation}
and thus, the equilibrium outcome is a fixed point of Equation~\ref{eqn:gran-fixedpoint} (see Figure~\ref{fig:granovetter}). This approach provides an explicit dynamics of influence spread, and characterizes the fraction of the population that is eventually influenced. Though the analysis seems reasonable at first glance, careful inspection reveals that \emph{the implied system dynamics will involve nodes resampling their thresholds at every timestep}. One should note that the threshold distribution is introduced to capture the variation in the unknown norms/preferences of the individuals in the population, but once sampled, they should remain unaltered during the process. Also, in Equation~\ref{eqn:gran-fixedpoint}, there is no distinction between nodes that are already active and the nodes that are still susceptible to influence. This contradicts the assumption that the spread of influence is progressive (where nodes once influenced will remain active until the end of the process).

On the other hand, Kempe et al. \cite{kempe-etal03max-spread-infl} assumed a uniform distribution for the influence threshold, and focused on the algorithmic problem of selecting an initial seed set (of a given size) so as to maximize the final influenced set. Although, the model has progressive spread dynamics, the evolution itself was not a primary concern in \cite{kempe-etal03max-spread-infl}. It was explicitly noted that the thresholds need to be sampled only once \emph{at the beginning of the process}.

Finally, for the uniform distribution of threshold (as assumed throughout in \cite{kempe-etal03max-spread-infl}), Equation~\ref{eqn:gran-fixedpoint} does not yield any useful insight, i.e., it does not predict a spread of influence.

Thus, although both Granovetter \cite{granovetter78threshold-models} and Kempe at al. \cite{kempe-etal03max-spread-infl} work with influence threshold distributions, the dynamics of the influence process are quite different. Further, although Kempe et al. work only with uniformly distributed threshold, Granovetter permits more general threshold distributions. The point of departure of our work is to adopt the idea of general threshold distributions from \cite{granovetter78threshold-models}, while retaining the more natural model of sampling each individual's threshold just once in the beginning, and of the progressive spread of influence from \cite{kempe-etal03max-spread-infl}.

\begin{figure}[t]
\centerline{\includegraphics[scale=0.4]{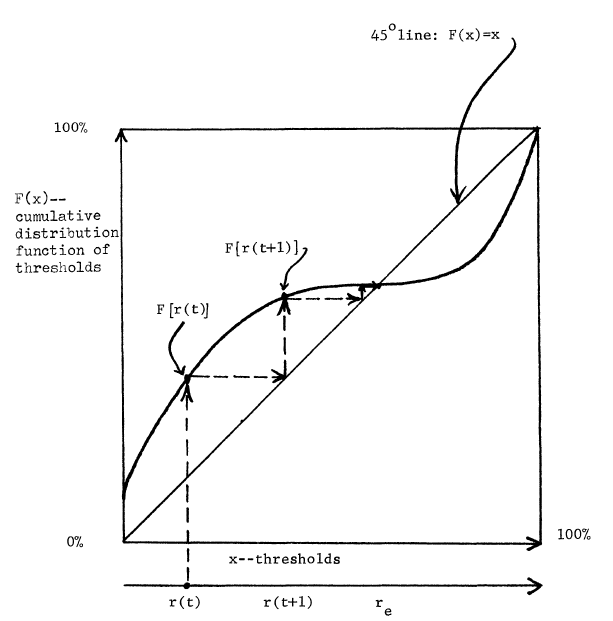}}
\caption{Granovetter's Fixed point dynamics. The intersection(s) of the cumulative distribution function with the $45^\circ$ line, indicate the equilibrium points, i.e. solution to the fixed point equation~\ref{eqn:gran-fixedpoint}. (Source \cite{granovetter78threshold-models})}
\label{fig:granovetter}
\end{figure}

\textit{Our Contributions:}
We adopt a fluid limit approach for analytically characterizing the dynamics of the spread of influence in the Linear Threshold model, with general threshold distributions. In order to do this, in Section~\ref{sec:math-model} we propose a homogeneous version of the Linear Threshold model called the Homogeneous Influence Linear Threshold (HILT) model, with arbitrary threshold distribution. In Section~\ref{sec:ode-approx} we characterize the evolution of influence in the HILT model as a Markov process, and using Kurtz's theorem \cite{kurtz70ode-markov-jump-processes}, we derive a system of ordinary differential equations (o.d.e.). We provide simulation results that show that the o.d.e. approximates the original process fairly well for large values of $N$, the population size. In Section~\ref{sec:uniform}, we explicitly solve the o.d.e. for the uniform threshold distribution thus providing an explicit characterization of the evolution of influence in the model of Kempe et al.\cite{kempe-etal03max-spread-infl}. 
We also provide an analytical expression for the terminal spread of influence and use it to address some optimization problems (Section~\ref{sec:numerical}). 

We note that the threshold distribution features in the o.d.e. via the \emph{hazard function}\cite{cox61renewal-theory}, commonly used in survival/failure analysis. To the best of our knowledge, this is the first work that incorporates the hazard function to characterize variation among the individuals in an epidemic model. The variation of risk as an epidemic progresses has been empirically observed in recent studies in veterinary medicine \cite{morton11epidemic} which has highlighted the need to ``consider possible differences in the risk of infection among subgroups in the population''.

We then proceed to study the effect of the threshold distribution in Section~\ref{sec:threshold}. We also note that under the exponential threshold distribution, the fluid dynamics of the HILT model is equivalent to that of the classic SIR model from epidemiology \cite{daley-gani99epidemic-modeling}, thus providing another interesting link between the influence spread and epidemics literature. Finally, we also show that the analysis can be extended to a heterogeneous system with communities (Section~\ref{sec:multiclass}), and conclude with some possible future directions. 

\textit{Comment on Network Topology:}
It has been noted that network topology plays a crucial role in the spread of influence on a social network \cite{gould93collective-network}. In this work, however, we consider only completely connected graphs, while deriving the fluid limit equations. While it is necessary to study the impact of degree distribution on the fluid dynamics, our primary focus is to provide the missing link between Kempe's model \cite{kempe-etal03max-spread-infl} and Granovetter's model \cite{granovetter78threshold-models}, and thus do not discuss the effect of network topology in this paper. 

\section{Mathematical Model}
\label{sec:math-model}
We shall first introduce the network model described in Kempe et al.\ \cite{kempe-etal03max-spread-infl}. The social network is a weighted directed graph $\mathcal{N}=(V,E)$, where the edge weight $w_{i,j}$ gives a measure of influence of node $i$ on node $j$. The activation process (see Figure~\ref{fig:spread_infl}) begins with an initial set of active, infectious nodes $\mathcal{A}_{0} = \mathcal{D}_{0}$ and takes place in discrete time steps. Each active node spreads its influence to \emph{each} of its inactive neighbours. By the activation process, some of the neighbours become activated to be part of $\mathcal{D}_1$, and can spread their influence in the next step. At the end of each step the population is partitioned into three sets of nodes: nodes that were just activated in that step $\mathcal{D}_k$ (also referred to as \emph{infectious} nodes), active nodes that have already exercised their influence $\mathcal{B}_k = \mathcal{A}_{k-1}$ and, hence, are no longer infectious, and the set of inactive 
nodes ($\mathcal{N}\backslash \mathcal{A}_k$). Note that $D_k \subseteq A_k$. The activation process stops at a random time $U$ when there are no more infectious nodes, i.e., $D_U = \emptyset$ and a \emph{terminal set} $A_{U}$ is reached, from where the activation process cannot proceed further. We also assume that once a node has become active, it cannot become inactive (\emph{progressive case}).

\subsection{Linear Threshold model}
An activation model describes how the infectious nodes cause the inactive nodes to become active (and infectious). There are two widely used activation models, namely, the \emph{Linear Threshold model} and the \emph{Independent Cascade model}, proposed in \cite{kempe-etal03max-spread-infl}. Our work in this paper begins with the Linear Threshold (LT) model. In the LT model, \( \sum_{i\neq j}  w_{i,j} \leq 1 \), i.e., the maximum possible influence on any node is bounded by 1 (see Figure~\ref{fig:sn-model}). In this model, each node $j$ randomly chooses a threshold $\Theta_{j}$ from a uniform distribution $U[0,1]$, \emph{at the beginning}. An inactive node, receives influence from all its active neighbours, and gets activated once the net received influence exceeds the chosen threshold. In other words, a node $j$ gets activated in step $k$ if, it had been inactive until step $k-1$, i.e. $ j \notin A_{k-1}$, and
\[ \sum_{i \in A_{k-1}}   w_{i,j} \geq \Theta_{j} \] 

\begin{figure}[t]
\centering
\includegraphics[scale=0.3]{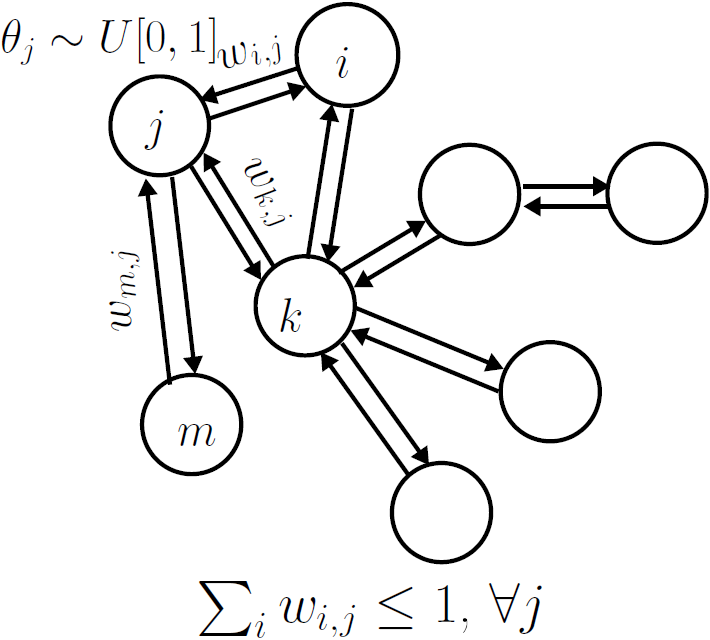}
\caption{An influence graph of a social network under the Linear Threshold model as introduced by Kempe et al. \cite{kempe-etal03max-spread-infl}.}
\label{fig:sn-model}
\end{figure}

\begin{figure}[t]
\centering
\includegraphics[scale=0.5]{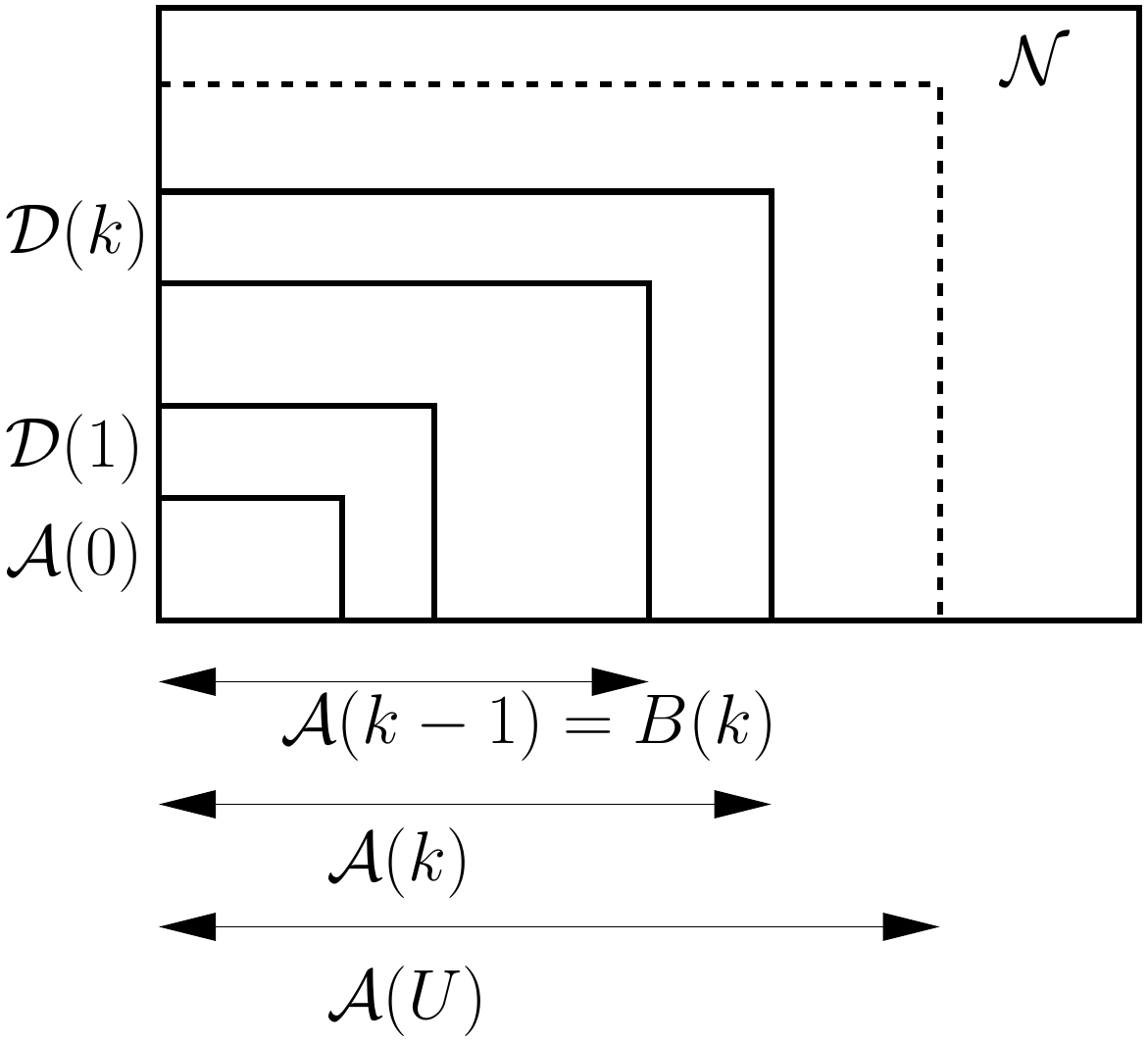}
\caption{Evolution of the set of influenced nodes under the Linear Threshold model.}
\label{fig:spread_infl}
\end{figure}

\subsection{HILT Network model}
Consider the population to be a social network $\mathcal{N}$ of $N$ nodes where the graph is complete, and each edge carries the same weight. Also, let the thresholds $\Theta_{j}$ be chosen from an arbitrary threshold distribution with cumulative density function (c.d.f.) $F$. We call this the Homogeneous Influence Linear Threshold (HILT) network model. The influence matrix $\mathbf{W}$ is given as follows:
for all $i \neq j$,
\[w_{i,j} = \gamma\] 
and
\[w_{i,i} = 0\]

\begin{figure}[t]
\centering
\includegraphics[scale=0.4]{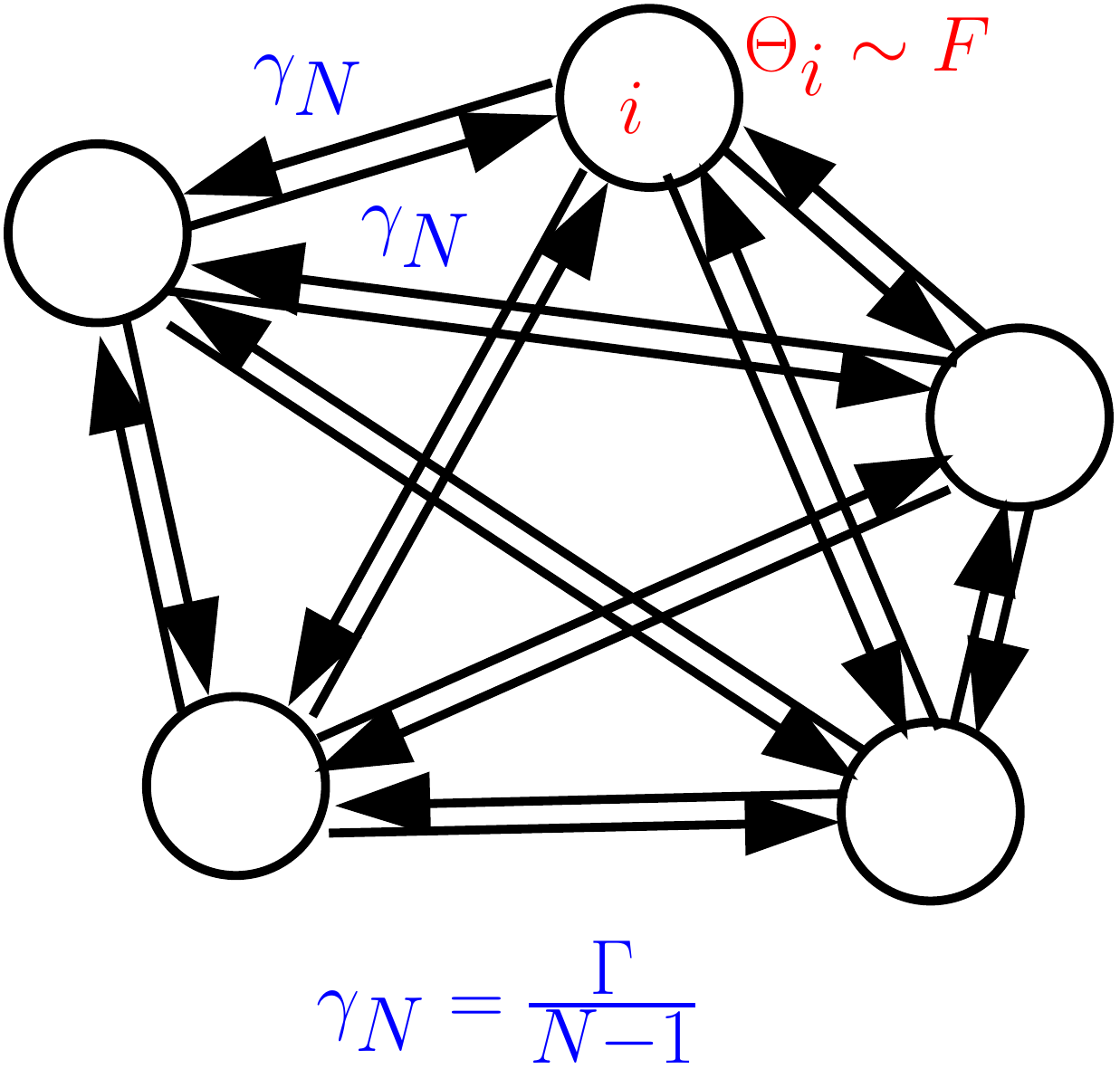}
\caption{The HILT Network model}
\label{fig:hilt}
\end{figure}

Carrying over the assumption in \cite{kempe-etal03max-spread-infl}'s model, we will assume that $\gamma_N \leq \frac{1}{N-1}$ when dealing with uniform threshold distribution. This is because, under uniform distribution, the maximum threshold is $1$, and it does not make sense to consider influences greater than $1$. In Section~\ref{sec:threshold}, when discussing various threshold distributions with unbounded support, we show that this restriction can be removed. 

\section{A Scaled Markov Chain and its Fluid Limit}
\label{sec:ode-approx}
Consider the HILT model on $N$ nodes, and with edge weights $\gamma (N)$ such that $\lim_{N \rightarrow \infty} \gamma (N) N = \Gamma$, and the threshold distribution at the nodes given by $F$. In this section we will use Kurtz's theorem \cite{kurtz70ode-markov-jump-processes} to obtain a two dimensional o.d.e. that can serve as a fluid approximation for the evolution of the stochastic processes in the HILT model. 

Let $A(k)$ and $D(k)$, respectively, be the \emph{sizes} of the active and infectious sets at time $k$. In the HILT model, due to homogeneity, the precise membership of these sets is irrelevant and it is sufficient to keep track of set sizes. Instead of $A(k)$, we will work with $B(k) = A(k-1)$ to distinguish the active nodes that have exercised their influence, and the infectious nodes. Recall that, $B(k)$ is the size of the subset of active nodes that \emph{have exercised their influence} by time $k$, whereas $D(k)$ is the size of the subset of active nodes at time $k$ that \emph{have not yet had a chance to exert their influence} on the inactive nodes. It is easy to observe that $(B(k), D(k))$ is a discrete time Markov chain (DTMC) (see Appendix~\ref{app:dtmc}). By definition,

\[B(k+1)= B(k) + D(k) \]

\subsection{A Scaled Markov Chain}
In order to obtain an approximating o.d.e., we need to work with an appropriately scaled Markov process $(B^{N}(k), D^{N}(k))$, which can be thought of as evolving on a time scale $N$ times faster than that of the original system. We can visualize this process as evolving over ``minislots'' of duration $1/N$, whereas the original process evolves at the epochs $0,1,2,\cdots$. Since this new process runs on a faster time scale, we need to slow down its dynamics. In each minislot, each node in $D^{N}(k)$ decides to spread its influence with probability $\frac{1}{N}$ or defer with probability $1-\frac{1}{N}$. In the former case, it contributes its influence of $\gamma$ and then moves to the set $B^{N}(k+1)$, else it stays in $D^{N}(k+1)$ set (see Figure~\ref{fig:minislots}). A similar scaling has been used in the context of the analysis of random multi-access algorithms by Bordenave et al.~\cite{bordenave05mean-fields}. The reason for such a scaling is explained in Appendix~\ref{app:scaling-kurtz}, where we 
contrast it with the traditional amplitude and time scaling. The evolution of this process can be written as follows:

\begin{figure}[t]
\centering
\includegraphics[scale=0.3]{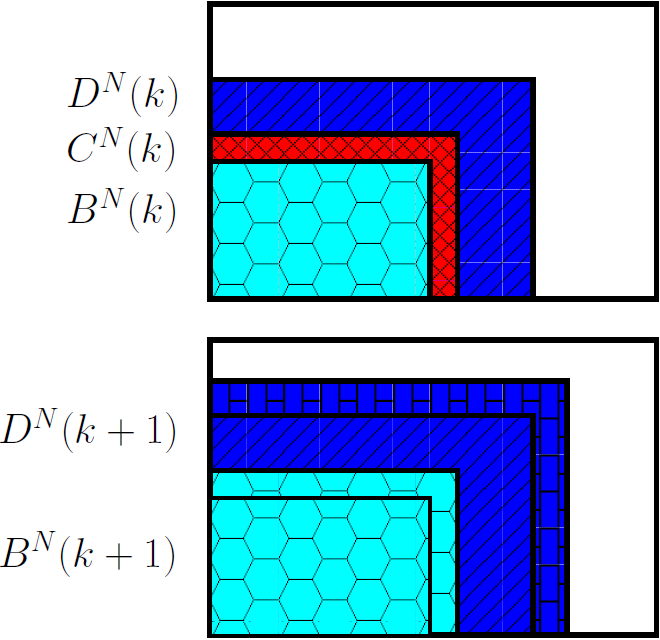}
\caption{Evolution of the scaled process}
\label{fig:minislots}
\end{figure}
\[ B^{N}(k+1) = B^{N}(k) + \frac{D^{N}(k)}{N} + Y^{N}(k+1) \]
\begin{eqnarray*}
D^{N}(k+1) &=& \frac{N-1}{N} D^{N}(k) + Z^{N}(k+1) \\
&+& \frac{F(\gamma(B^{N}(k) + \frac{D^{N}(k)}{N}))- F(\gamma(B^{N}(k))}{1-F(\gamma(B^{N}(k))} \times \\
&& \hspace{2cm} (N-B^{N}(k) -D^{N}(k)) \\
\end{eqnarray*}
where $Y^{N}(k+1)$ and $Z^{N}(k+1)$ are zero mean random variables. 

Dividing the evolution equations by $N$ (the number of nodes in the network) and defining $\Tilde{B}^{N}(k) = \frac{B^{N}(k)}{N}$, $\Tilde{D}^{N}(k) = \frac{D^{N}(k)}{N}$, 
we can obtain the drifts for $\Tilde{B}^{N}(k), \Tilde{D}^{N}(k)$ for the fraction of nodes in each state. 
\[ \Tilde{B}^{N}(k+1) = \Tilde{B}^{N}(k) + \frac{\Tilde{D}^{N}(k)}{N} + \Tilde{Y}^{N}(k+1) \]
\begin{eqnarray*}
\lefteqn{\Tilde{D}^{N}(k+1)}\\
&=& \frac{N-1}{N} \Tilde{D}^{N}(k) + Z^{N}(k+1) \\
&+& \frac{F(\gamma(\Tilde{B}^{N}(k) + \frac{\Tilde{D}^{N}(k)}{N}))- F(\gamma(\Tilde{B}^{N}(k))}{1-F(\gamma(\Tilde{B}^{N}(k))} \times \\
&& \hspace{2cm} (N-\Tilde{B}^{N}(k) - \Tilde{D}^{N}(k)) \\
\end{eqnarray*}

Let $f_1^{N}(\Tilde{B}^{N}(k),\Tilde{D}^{N}(k))$ and $f_2^{N}(\Tilde{B}^{N}(k),\Tilde{D}^{N}(k))$ denote the mean drifts of $\Tilde{B}^{N}(k), \Tilde{D}^{N}(k)$. 

Consider the limiting drift function of $f_2^{N}(.)$ and observe that, 
\begin{eqnarray*}
\lefteqn{\lim_{N \rightarrow \infty} N \bigg( \frac{F(x+ \frac{y}{N}) - F(x)}{1-F(x)} \bigg)}\\
&=& \lim_{N \rightarrow \infty} \frac{y}{1-F(x)} \frac{F(x+ \frac{y}{N}) - F(x)}{y/N} \\
&=& \frac{y f(x)}{1-F(x)} \\
\end{eqnarray*}
Now consider $f_1(b,d)= d$ and $f_2(b,d) =  \frac{f(\Gamma b)\Gamma d}{1- F(\Gamma b)} (1-b-d) - d$ and define,
\[ f(b,d) := \bigg( f_1(b,d), f_2(b,d) \bigg) \]
\begin{theorem}
\label{thm:kurtz-theorem-hilt}
Given the Markov process $( \Tilde{B}^{N}(k), \Tilde{D}^{N}(k) )$, we have for each $T > 0$ and each $\epsilon > 0$,

\begin{eqnarray}
 P \bigg( \sup_{0 \leq t \leq T} \big| \big| \big( \Tilde{B}^{N}( \lfloor Nt \rfloor ),\Tilde{D}^{N}( \lfloor Nt \rfloor ) \big) - \big( b(t),d(t) \big) \big| \big| > \epsilon \bigg) &&  \nonumber \\
 & & \hspace{-2cm} \stackrel{N\rightarrow \infty}{\rightarrow} 0 \nonumber 
\end{eqnarray}
where $(b(t),d(t))$ is the unique solution to the ODE,

\[ \dot{b} = d\]
\[ \dot{d} = \frac{f(\Gamma b) \Gamma d}{1- F(\Gamma b)} (1-b-d) - d \]
with initial conditions $(b(0)=0,d(0)=a(0))$.

\end{theorem}

\begin{proof}
This is essentially an instance of Kurtz's theorem \cite{kurtz70ode-markov-jump-processes}; Also see \cite{darling02limits-purejump-markov}. In Appendix~\ref{app:theorem1} we provide the statement of Kurtz's theorem for our context, and the verify the necessary conditions to guarantee the convergence of the Markov processes to the fluid limit o.d.e. .
\end{proof}

\textit{Remark:} We know that the hazard function corresponding to the c.d.f. $F(x)$ is given by 
\[h_F(x) = \frac{f(x)}{1-F(x)}\] and hence the o.d.e. becomes,
\begin{equation}
 \dot{b} = d
 \label{eqn:b_limit}
\end{equation}
\begin{equation}
 \dot{d} = h_F(\Gamma b) \Gamma d  (1-b-d) - d 
 \label{eqn:d_limit}
\end{equation}

\subsection{Accuracy of the o.d.e. approximation}
Figure~\ref{fig:hilt-convergence} shows the convergence of the scaled process to the o.d.e. with increasing network sizes $N=50,100,500,1000$ and for $\Gamma=0.9$ and $d_0 =0.2$. We observe that for $N=1000$ the o.d.e. approximates the scaled process fairly well.

As noted in Appendix~\ref{app:scaling-kurtz}, the probabilistic scaling does not exactly replicate the original process. Hence, we have also compared the evolution of the original unscaled process for a fixed value of $N=1000$, with the o.d.e approximation. The results are shown in Figure~\ref{fig:kurtz-errorbar}, where multiple sample paths of the original process (obtained by using different random number seeds) are plotted along with the (deterministic) o.d.e. solution. We find that the o.d.e. solution approximates the mean evolution of the original process well.

\begin{figure}[t]
\centerline{\includegraphics[scale=0.17]{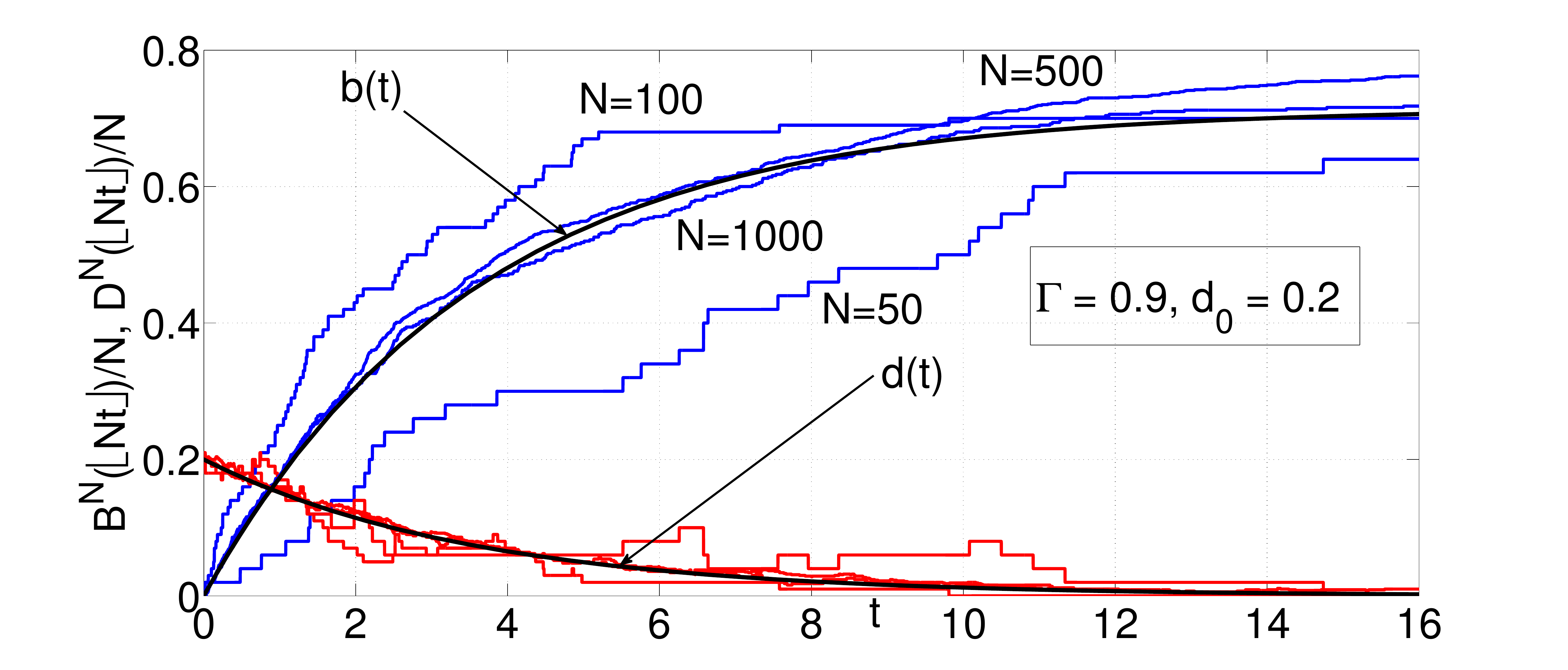}}
\caption{Trajectory of the fluid limit $(b(t),d(t))$ plotted along with samplepaths of the scaled process $\Tilde{B}^{N}( \lfloor Nt \rfloor ),\Tilde{D}^{N}( \lfloor Nt \rfloor )$ for $N=50,100,500,1000$.}
\label{fig:hilt-convergence}
\end{figure}

\begin{figure}[t]
\centerline{\includegraphics[scale=0.17]{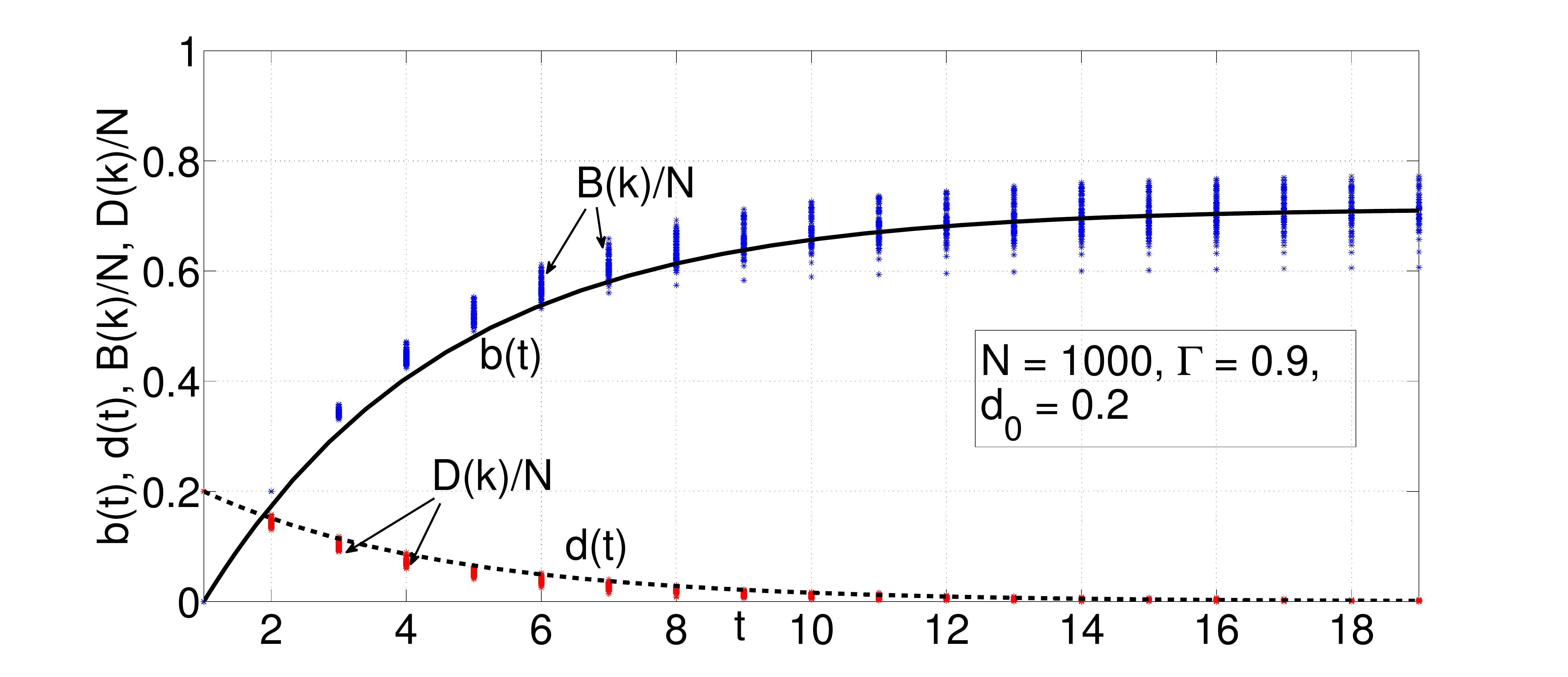}}
\caption{Trajectory of the fluid limit $(b(t),d(t))$ plotted along with multiple runs of the original process (normalized) $(\frac{B(k)}{N},\frac{D(k)}{N})$ for $N=1000$. }
\label{fig:kurtz-errorbar}
\end{figure}

\section{Uniform Threshold Distribution}
\label{sec:uniform}
In this section, we will consider the o.d.e. approximation of the HILT process, under the uniform distribution of threshold. The hazard function for uniform distribution is given by $h_F(x) = \frac{1}{1-x}$ and thus the system of o.d.e. becomes, 
\[ \dot{b} = d\]
\[ \dot{d} = -d + \frac{\Gamma d}{1 - \Gamma b} (1 -b - d) \]
It turns out that we can explicitly solve the above system, thus yielding closed form expressions for $(b(t), d(t))$. We will derive these closed form expressions and use these explicit expressions. 

\subsection{Solution to the o.d.e. }
On solving the o.d.e. for uniform distribution, with initial conditions $b(0)=0,d(0)=d_0$ and defining $r = 1 - \Gamma + \Gamma d_0$, we get
\[b(t) = \frac{d_0}{r} - \frac{d_0}{r} e^{-rt} \]
\[d(t) = d_0 e^{-rt} \]
In Appendix \ref{app:ode-solving} we provide the steps involved in obtaining the solution. From the above equations we can state the following theorem:

\subsection{Terminal spread of influence}
\label{subsec:terminal-uniform}
The following theorem results from a simple observation of the o.d.e.'s.
\begin{theorem}
Given that we start with $d_0$ fraction of nodes in the infectious set in an HILT network with parameter $\Gamma$, then the final fraction of activated nodes will be $\frac{d_0}{r}$ where $r=1-\Gamma+\Gamma d_0$.
\end{theorem}

\textit{Remarks:}
\begin{itemize}
 \item We might also be interested in the question of choosing the right $d_0$ which can give us the required $b_\infty$, and we see that
\[ d_0 = \frac{b_\infty (1 -\Gamma)}{1 - b_\infty \Gamma} \]  
 \item We observe that, for large $N$, as long as $\Gamma < 1$ we cannot influence the entire population (i.e., $b_\infty =1$) unless we start off with the entire population active (i.e., $d_0 =1$). But if $\Gamma = 1$ then $b_\infty =1$ provided $d_0 > 0$.
\end{itemize}

Consider the discrete influence process (the Kempe model \cite{kempe-etal03max-spread-infl}), and let $\sigma^{(\mathcal{N},\mathcal{A}_0)} = \mathbb{E}^{(\mathcal{N},\mathcal{A}_0)}[|A_U|]$ be the expected size of the \emph{terminal set} $A_U$, starting with $\mathcal{A}_0$ as the initial set in the network $\mathcal{N}$. Since all initial sets are equivalent in the HILT model, we will be interested in the influence of a set of size $m$. Define $h_{\gamma}^{(N)}(m) :=\sigma^{(\mathcal{N},\mathcal{A})}$, for all $\mathcal{A}$ of size $m$. 

By using results from \cite{srini-kumar11LT-model-ncc}, we can show that, 

\[h_{\gamma}^{(N)}(m) = m[1 + (N-m)\gamma [ 1 + (N-m-1)\gamma[ 1 + \cdots \]

The behaviour of $h_{\gamma}^{(N)}(m)$ as a function of $\gamma_N$ and $m$ can be seen in Figure~\ref{fig:hilt_k_3000} (depicted by solid lines), for a network of 3000 nodes. We also superimpose the behavior of $b_\infty$ against $d_0$ (depicted by asterisks). We observe that there is an exact match, except for $\Gamma = 1$. For $\Gamma =1$, as seen earlier, we know that $b_\infty =1$ as long as $d_0 > 0$. This is however true only in the fluid limit, and hence the discrepancy for finite $N$.

\begin{figure}[t]
\centerline{\includegraphics[scale=0.20]{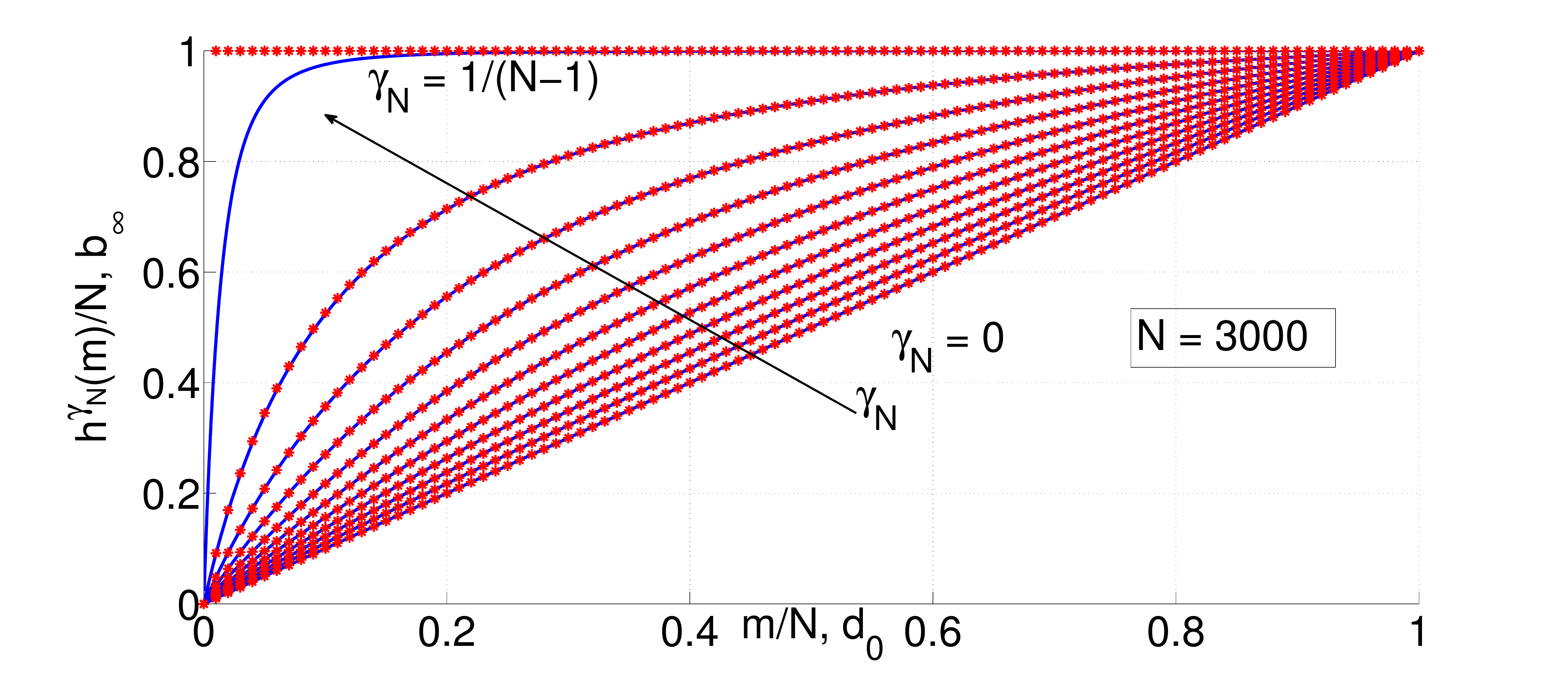}}
\caption{$h_{\gamma}^{(N)}(k)$ versus $k$ for $N=3000$ (shown by solid lines) and $b_\infty$ versus $d_0$ (shown by asterisks) for various values of $\gamma_N$. }
\label{fig:hilt_k_3000}
\end{figure}

Taking $\Gamma = N \gamma$ and $d_0 = \frac{m}{N}$, we can show that as $N \rightarrow \infty$, $\frac{h_{\gamma}^{(N)}(m)}{N} \rightarrow b_\infty$. See Appendix \ref{app:hilt-convergence} for the proof. This provides another verification of the accuracy of the o.d.e. approximation for large $N$. 

\section{Time constrained optimization}
\label{sec:numerical}
While the analytical expression $h_{\gamma}^{N} (m)$ derived earlier for HILT gives only the expected size of the terminal set, the o.d.e. dynamics approximates the trajectory of influence evolution, for large $N$. This can be useful, especially in problem settings where the time taken by the process for the spread of influence is also considered, in addition to the size of the initial set.

\begin{theorem}
Given the initial fraction of infected nodes $d_0$ in an HILT network with parameter $\Gamma$, the time we have to wait to get at least $\alpha$ ($\alpha < \frac{d_0}{r}$) fraction of nodes active is given by,

\[ T(\alpha, d_0, \Gamma) = \frac{1}{r} \ln \bigg( \frac{1-r}{1-\frac{\alpha}{d_0}r} \bigg) \]
where $r = 1 - \Gamma + \Gamma d_0$.

\end{theorem}

\begin{proof}
Firstly, note that since $a_\infty = b_\infty = \frac{d_0}{r}$, $\alpha$ must be less than $\frac{d_0}{r}$. Since we are observing the process at a finite time $T$, $d(T)$ is not zero. Hence, we should look at the value of $a(T)=b(T)+d(T)$ and set it to $\alpha$. We get,

\[a(T) = d_0 \big( \frac{1}{r} - ( \frac{1}{r}-1) e^{-rT} \big) = \alpha\]

Rearranging terms,we get the expression for $T(\alpha, d_0, \Gamma)$. 
\end{proof}

A more interesting question would be to determine the $d_0$ to be chosen so that by time $T$ we will have at least $\alpha$ fraction of the nodes activated, in the HILT network with parameter $\Gamma$. Unfortunately, we will not be able to get a closed form expression for this, and it can be solved numerically using the following fixed point equation.

\[ e^{-rT} = \frac{1-\frac{\alpha}{d_0} r}{1-r} \]

We can use the \emph{iterative bisection method} obtain the fixed point of the above equation. Let $F(d_0) = e^{-rT}$ and $G(d_0)=\frac{1-\frac{\alpha}{d_0} r}{1-r}$. We know that $d_0^{\star}$ that solves $F(d_0)=G(d_0)$ will lie in $[\frac{\alpha (1 -\Gamma)}{1 - \alpha \Gamma}, 1]$ and that the solution is unique, since $a(T)$ is a monotonic function in $d_0$. We also know that for $d_0 < d_0^{\star}$, $F(d_0) > G(d_0)$ and for $d_0 > d_0^{\star}$, $F(d_0) < G(d_0)$. 

Under the above conditions, we find that the bisection method will converge to $d_0^{\star}$. This is shown as Algorithm \ref{algo:bisection}. The method is illustrated in Figure~\ref{fig:bisection} for parameters $\Gamma = 0.8$, $\alpha = 0.7$, $T = 15$.

\begin{algorithm}
\LinesNumbered
$d_0^{min} = \frac{\alpha (1 -\Gamma)}{1 - \alpha \Gamma}$\;
$d_0^{max} = 1$\;
\While{1}
{
$x = (d_0^{min} + d_0^{max})/2$\;
\eIf{$F(x) - G(x) > 0$}
{
 $d_0^{min} = x$\;
}
{
$ d_0^{max} = x$\;
}
\If{$|F(x) - G(x)| < \epsilon$}
{
break\;
}
}
$d_0^{\star} = x$\;
\caption{Iterative Bisection method}
\label{algo:bisection}
\end{algorithm}

\begin{figure}[t]
\centerline{\includegraphics[scale=0.3]{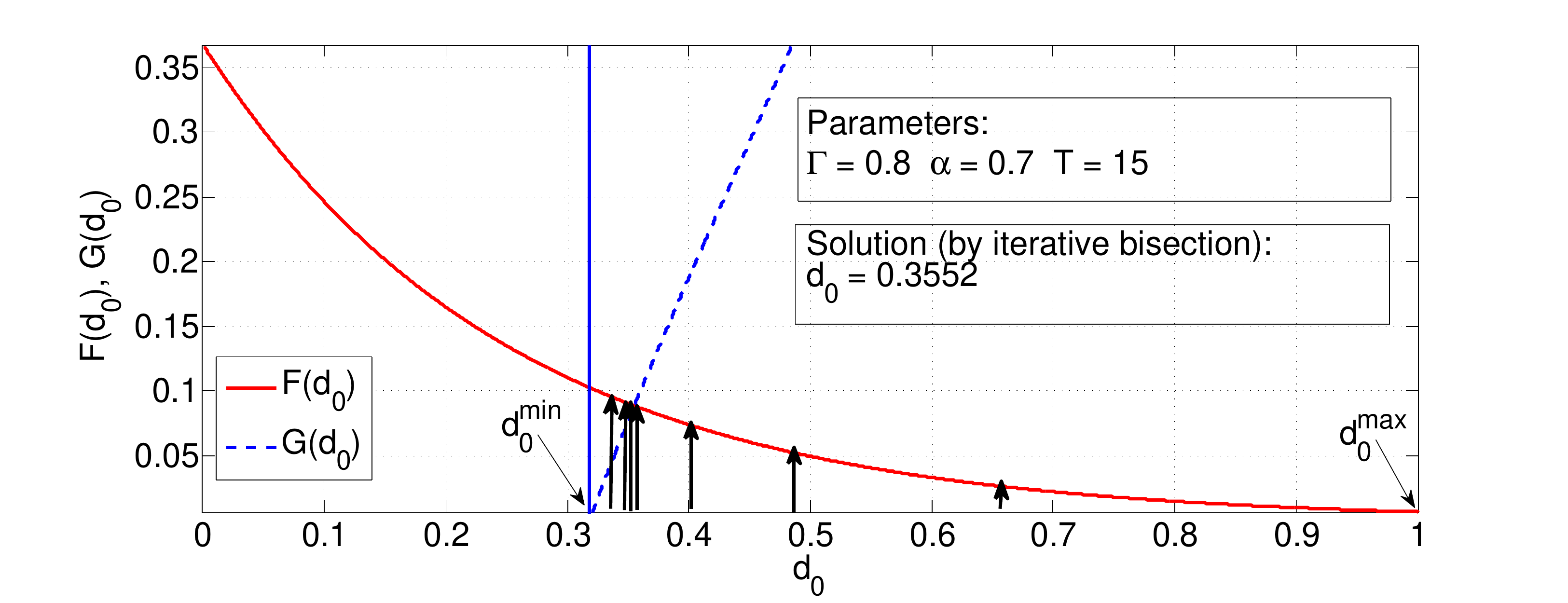}}
\caption{Evaluating $d_0$ by Iterative Bisection Method} 
\label{fig:bisection}
\end{figure}

The variation of $d_0^{\star}$ with respect to the parameters $\alpha$, $\Gamma$ and $T$ can be seen in Figures~\ref{fig:d_0_fixed_alpha_varying_Gamma_across_T},\ref{fig:d_0_fixed_Gamma_varying_alpha_across_T},\ref{fig:d_0_fixed_T_varying_Gamma_across_alpha}. 

In Figure~\ref{fig:d_0_fixed_alpha_varying_Gamma_across_T}, note that for $\Gamma = 0$, as expected, $d_0 = \alpha$, i.e., since there is no social interaction ($\Gamma = 0$), our terminal spread of influence will be equal to the initial seeding. Also, note that as the target time $T$ is reduced, we require higher values of $d_0$ to achieve the same $\alpha$ (for $T=0$, $d_0 = \alpha$). Finally, for $\Gamma =1$, $d_0$ asymptotically approaches $0$ for large $T$. From Figure~\ref{fig:d_0_fixed_Gamma_varying_alpha_across_T}, we see that as $\alpha$ increases, for a given $T$, the required $d_0$ monotonically increases. Finally, Figure~\ref{fig:d_0_fixed_T_varying_Gamma_across_alpha}, shows that the $d_0 \ \mathtt{vs} \ \alpha$ behavior for $T=8$ is qualitatively similar to the one in Figure~\ref{fig:hilt_k_3000} depicting $d_0 \ \mathtt{vs} \ b_\infty$.

\begin{figure}[t]
\centerline{\includegraphics[scale=0.17]{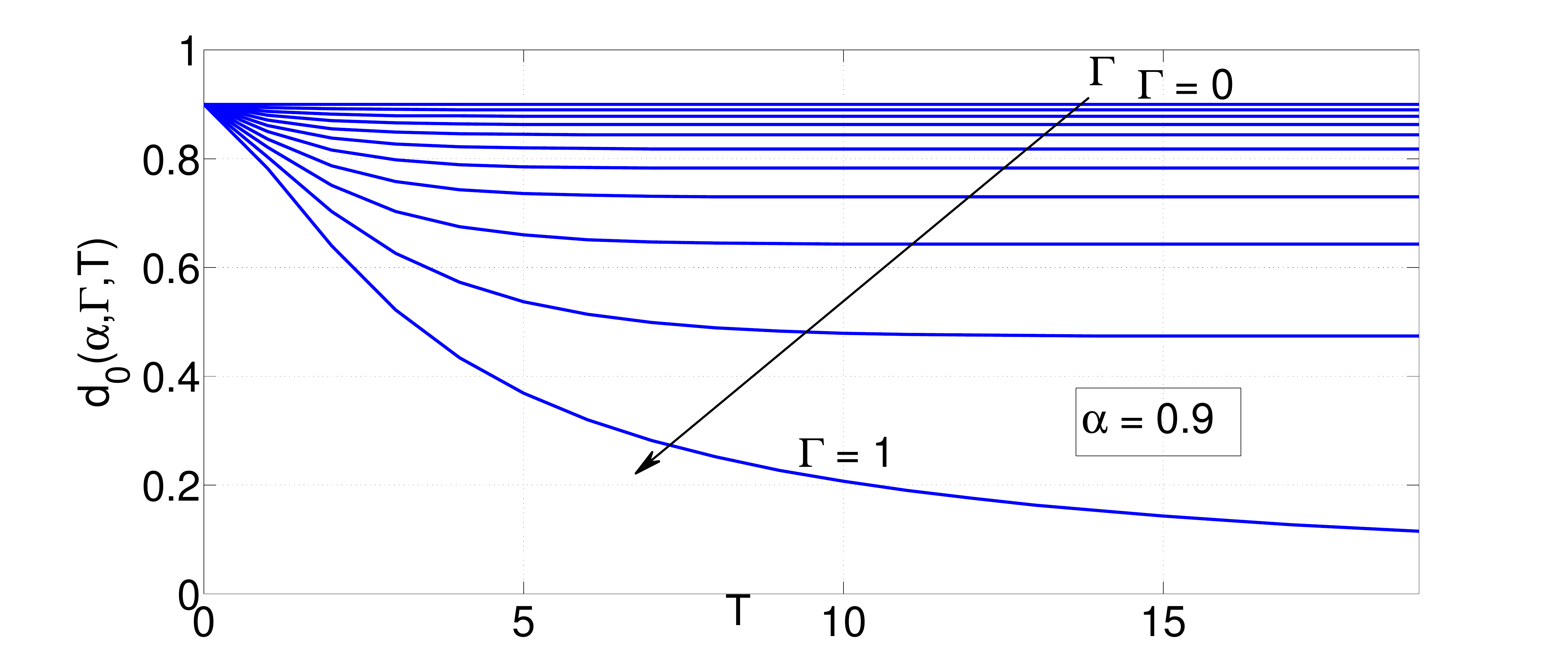}}
\caption{Variation of $d_0^{\star}$ across $T$ for various values of $\Gamma$ with $\alpha=0.9$ } 
\label{fig:d_0_fixed_alpha_varying_Gamma_across_T}
\end{figure}

\begin{figure}[t]
\centerline{\includegraphics[scale=0.17]{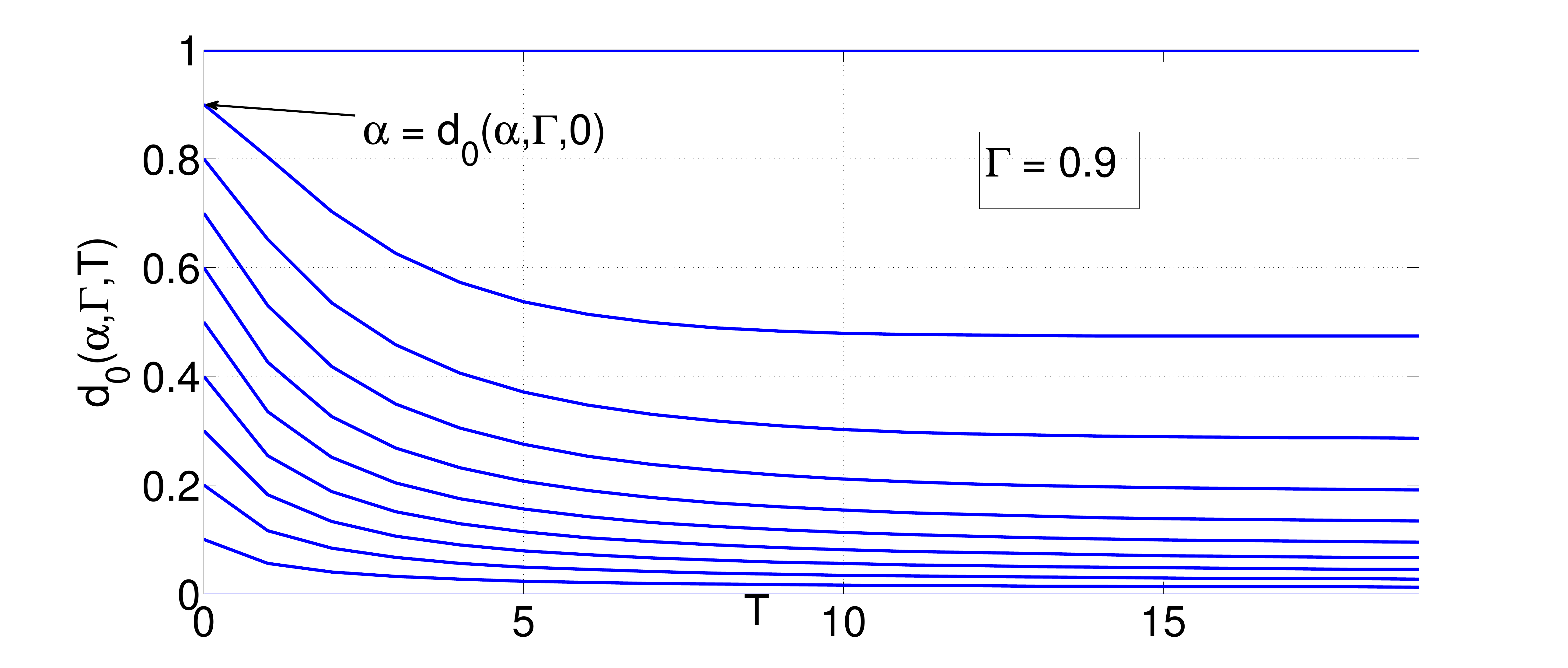}}
\caption{Variation of $d_0^{\star}$ across $T$ for various values of $\alpha$ with $\Gamma=0.9$}
\label{fig:d_0_fixed_Gamma_varying_alpha_across_T}
\end{figure}

\begin{figure}[t!]
\centerline{\includegraphics[scale=0.17]{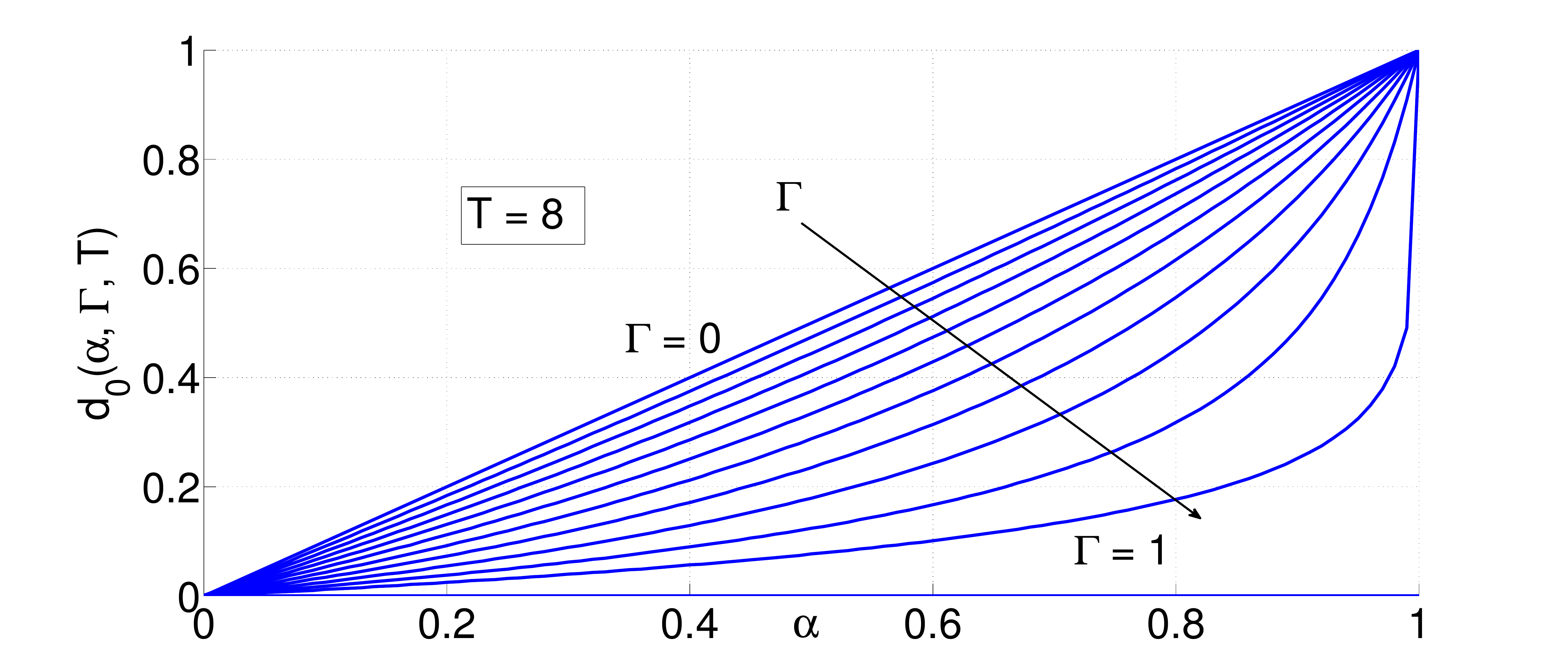}}
\caption{Variation of $d_0^{\star}$ across $\alpha$ for various values of $\Gamma$ at $T=8$}
\label{fig:d_0_fixed_T_varying_Gamma_across_alpha}
\end{figure}

\section{Effect of the Threshold distribution}
\label{sec:threshold}
Recall that the evolution of influence is given by:
\[ \dot{b} = d\]
\[ \dot{d} = h_F(\Gamma b) \Gamma d  (1-b-d) - d \]
Note that the evolution depends on the distribution of threshold via its hazard function,
\[ h(x) = \frac{f(x)}{1-F(x)}\] 
where $f(\cdot)$ and $F(\cdot)$ are the probability density and cumulative distribution functions of the threshold distribution, respectively. Hazard functions are widely used in failure/survival analysis. In this section, we will consider threshold distributions with different hazard function characteristics, and study the spread of influence.

As indicated earlier, the o.d.e. derived is valid for any $\Gamma > 0$, and in this Section, we will also consider cases when $\Gamma > 1$, while discussing threshold distributions with unbounded support. However, for uniform threshold distribution, we will restrict $\Gamma \leq 1$, since under this case $h_F(x) = \frac{1}{1-x}$, valid only for $x \in [0,1]$. 

\subsection{Exponential distribution}
Exponential distribution is widely used in scenarios where there is need for a constant hazard rate. This is also due to the fact that exponential distribution is the only memoryless continuous distribution. Consider the threshold $\theta_i$ distributed as exponential with parameter $\lambda$. We have 
\[f(x;\lambda)  = \lambda e^{-\lambda x}, \ x \geq 0 \]
\[F(x;\lambda)  = 1 - e^{-\lambda x}, \ x \geq 0 \]
Thus we get $h_F(x) = \lambda$. Plugging this in the o.d.e. expression we get, 
\[ \dot{b} = d\]
\[ \dot{d} = -d + \lambda \Gamma d (1 -b - d) \]
Observe that the above system of o.d.e. is equivalent to the dynamics of an SIR \emph{(Susceptible-Infective-Recovered)} epidemic, with infection rate $\lambda \Gamma$ and recovery rate $1$ \cite{daley-gani99epidemic-modeling}. The $b(t)$ and $d(t)$ processes respectively are equivalent to the \emph{Recovered} and \emph{Infective} processes of the SIR epidemic model. Thus we see that the under exponential distribution of threshold, the Linear Threshold model, in its fluid limit, is equivalent to a special case of the SIR model. This equivalence provides a hitherto undocumented link between influence spread models from viral marketing literature (Linear Threshold model) and a traditional epidemic model (SIR model). 

Figures~\ref{fig:uniform_exponential} and \ref{fig:uniform_exponential_large_gamma} compare the influence evolution under uniform and exponential distribution of threshold. Note that for the same mean threshold ($\mathbf{E}\theta = 0.5$) and smaller value of $\Gamma$ (Figure~\ref{fig:uniform_exponential}), exponential case yields a larger terminal influence spread. This is because, under the exponential distribution, there are more nodes with threshold close to zero. This also explains the steeper increase of $b(t)$ for exponential distribution compared to the uniform distribution case. In fact, from the respective o.d.e.s it is clear that $\dot{a}(0)=\dot{b}(0)+\dot{d}(0)$ for uniform distribution, is half that of exponential distribution with the same mean. 

But, for larger values of $\Gamma$ (Figure~\ref{fig:uniform_exponential_large_gamma}), uniform distribution yields a larger terminal influence spread. This is because, in the uniform case, the thresholds are bounded above by $1$, while in the exponential case, the support set for thresholds is unbounded. Thus, under the uniform distribution, as $\Gamma$ approaches $1$, the terminal spread of influence approaches $1$ (as noted in Section~\ref{subsec:terminal-uniform}).

\begin{figure}[ht!] 
  \subfigure[small $\Gamma$ regime]{ 
    \includegraphics[scale=0.17]{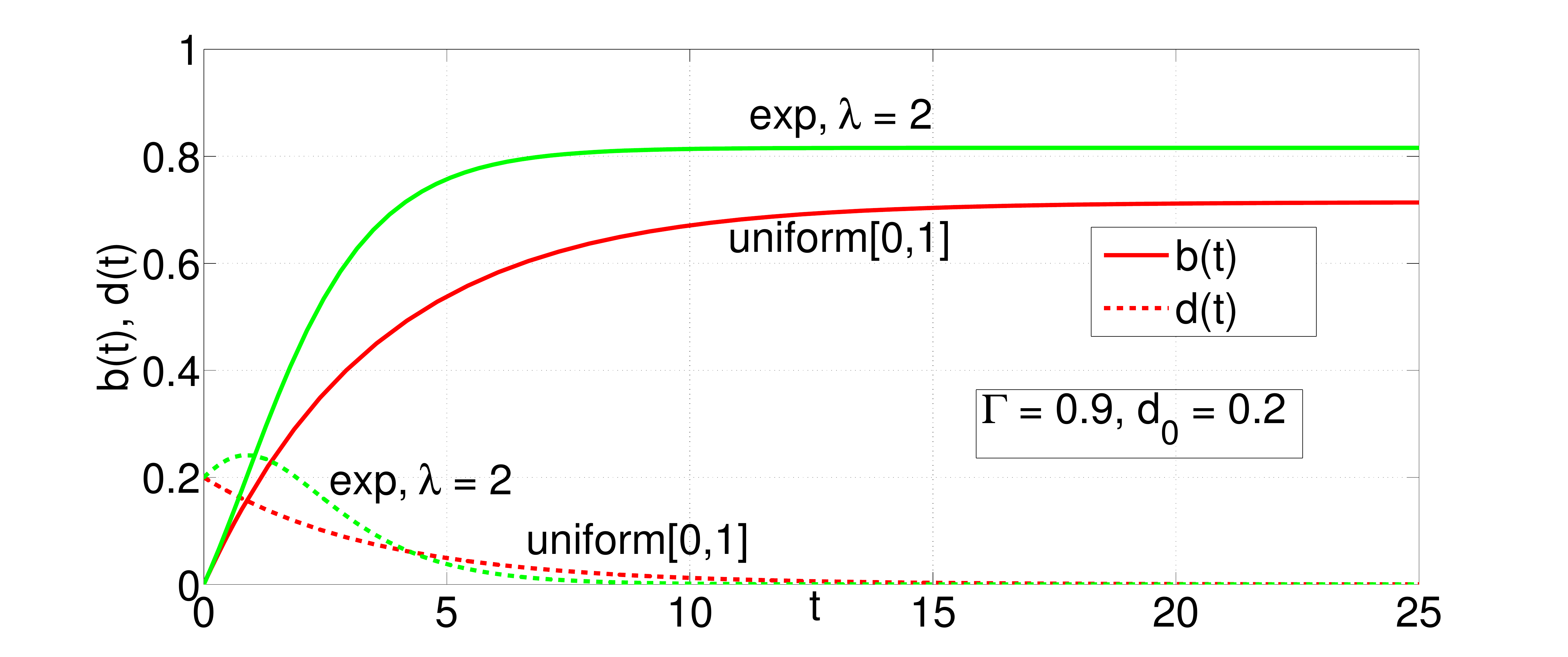} \label{fig:uniform_exponential}
  } 
  \quad 
  \subfigure[large $\Gamma$ regime]{
    \includegraphics[scale=0.17]{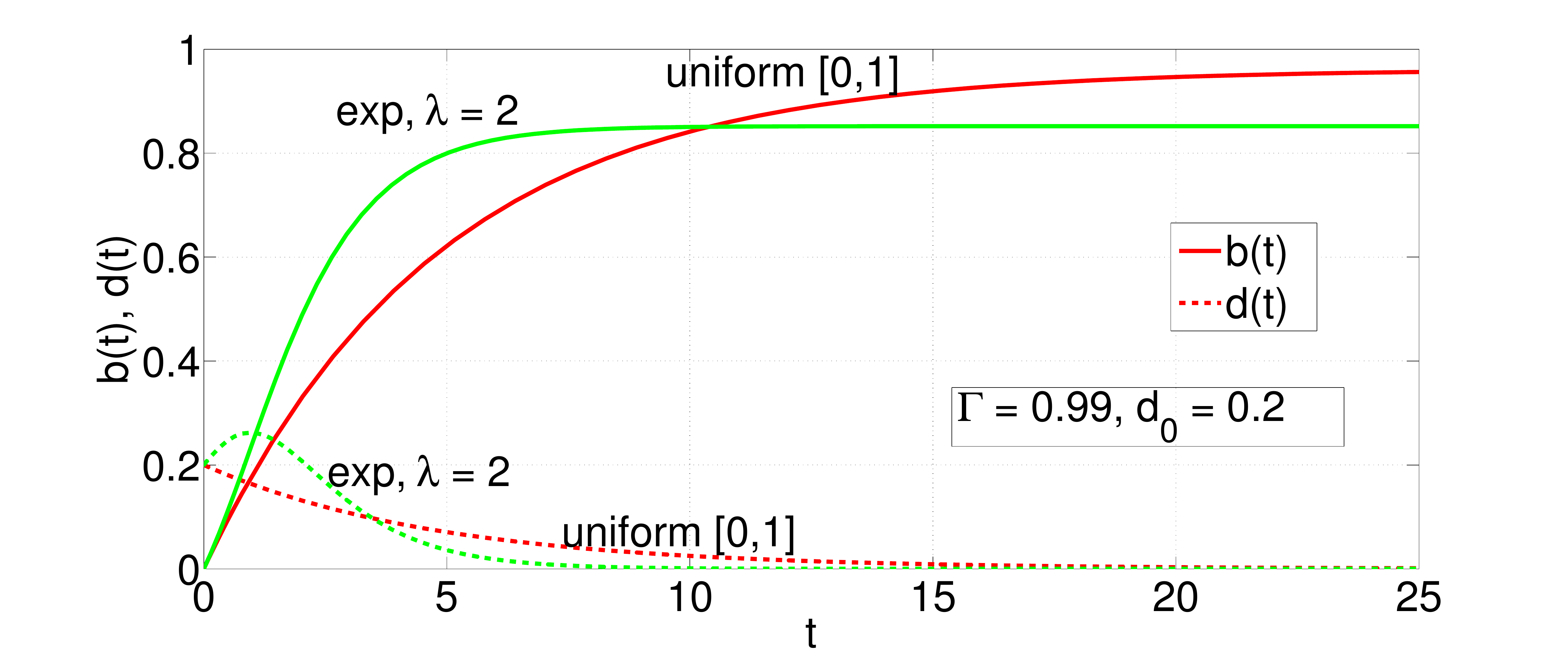} \label{fig:uniform_exponential_large_gamma}
  } 
  \caption{Comparison of influence spread between Uniform threshold distribution and Exponential threshold distribution with the same mean. We still use $\Gamma < 1$, since we are dealing with the uniform distribution} 
\end{figure}

\subsection{Weibull distribution}
Another distribution which is widely used in survival analysis is the Weibull distribution. The probability density function of a Weibull random variable is given by, 
\[ f(x: \lambda, k) = \frac{k}{\lambda} \bigg(\frac{x}{\lambda} \bigg)^{(k-1)} e^{-(\frac{x}{\lambda})^k} , \ x\geq 0 \]
If the random variable $X$ is the time to failure, then under the Weibull distribution,  the failure rate is proportional to a power of time. The hazard function is given by,
\[ h(x;  \lambda, k) = \frac{k}{\lambda} \bigg(\frac{x}{\lambda}\bigg)^{(k-1)}, \ x\geq 0 \]

In the above expression $\lambda$ is often referred to as the \emph{scale} parameter and $k$ is referred to as the \emph{shape} parameter. Figure~\ref{fig:weibull_f} shows the probability density function of Weibull distribution for various values of $k$. Note that for $k>1$, there are significantly high number of users with higher values of threshold, i.e., less susceptible to the spread of influence. The hazard rate for Weibull distribution can be increasing, constant or decreasing depending on the value of $k$. This is demonstrated in the Figure~\ref{fig:weibull_h}.

\begin{itemize}
 \item $k<1$ leads to decreasing hazard rate. This implies that nodes are less likely to become activated by an instantaneous influence, as the existing influence (which failed to activate the node) on them increases. 
 \item $k=1$ yields constant hazard rate, and in that case Weibull distribution is just the exponential distribution.
 \item $k>1$ yields an increasing hazard rate, which implies nodes are more likely to become activated by an instantaneous influence, as the existing influence on them increases.
\end{itemize}

\begin{figure}[ht!] 
  \subfigure[Probability density function]{ 
    \includegraphics[scale=0.17]{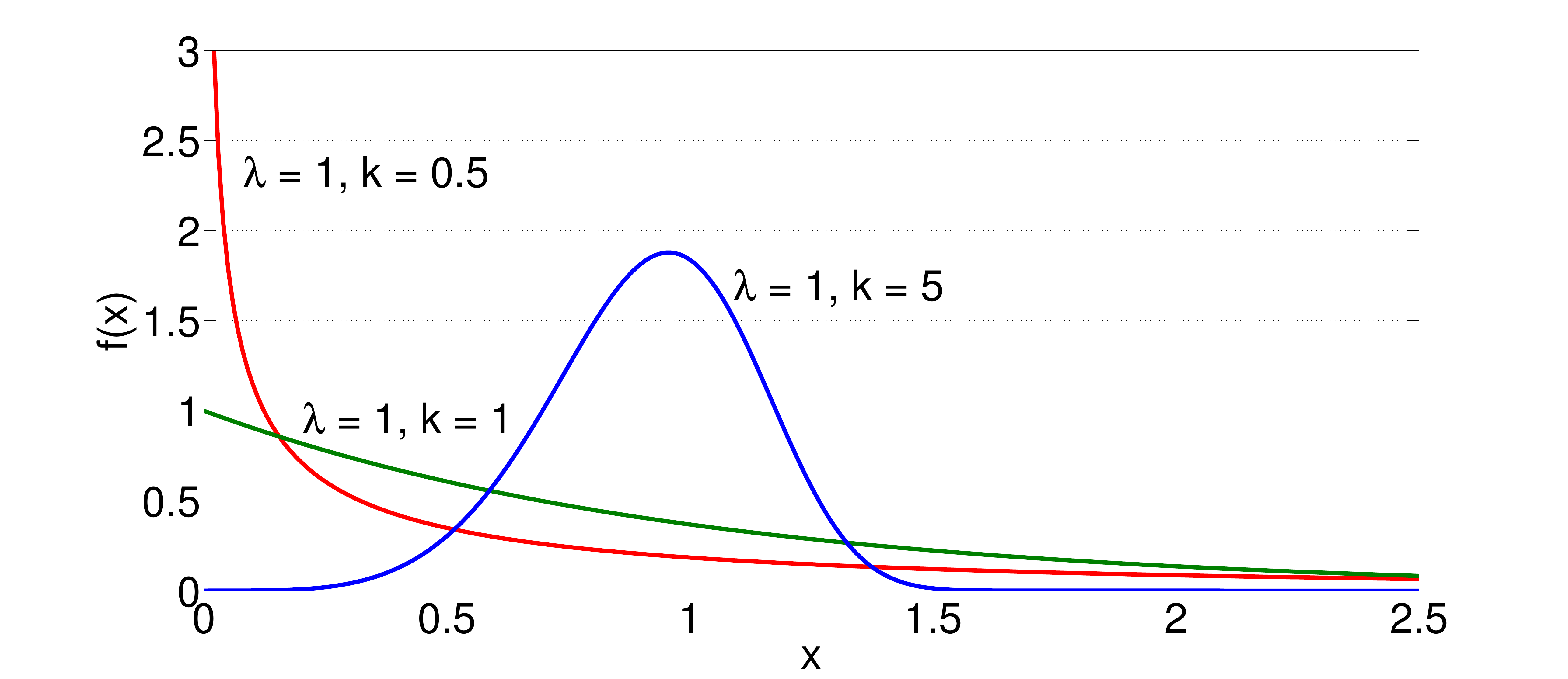} \label{fig:weibull_f}
  } 
  \quad 
  \subfigure[Hazard function]{
    \includegraphics[scale=0.17]{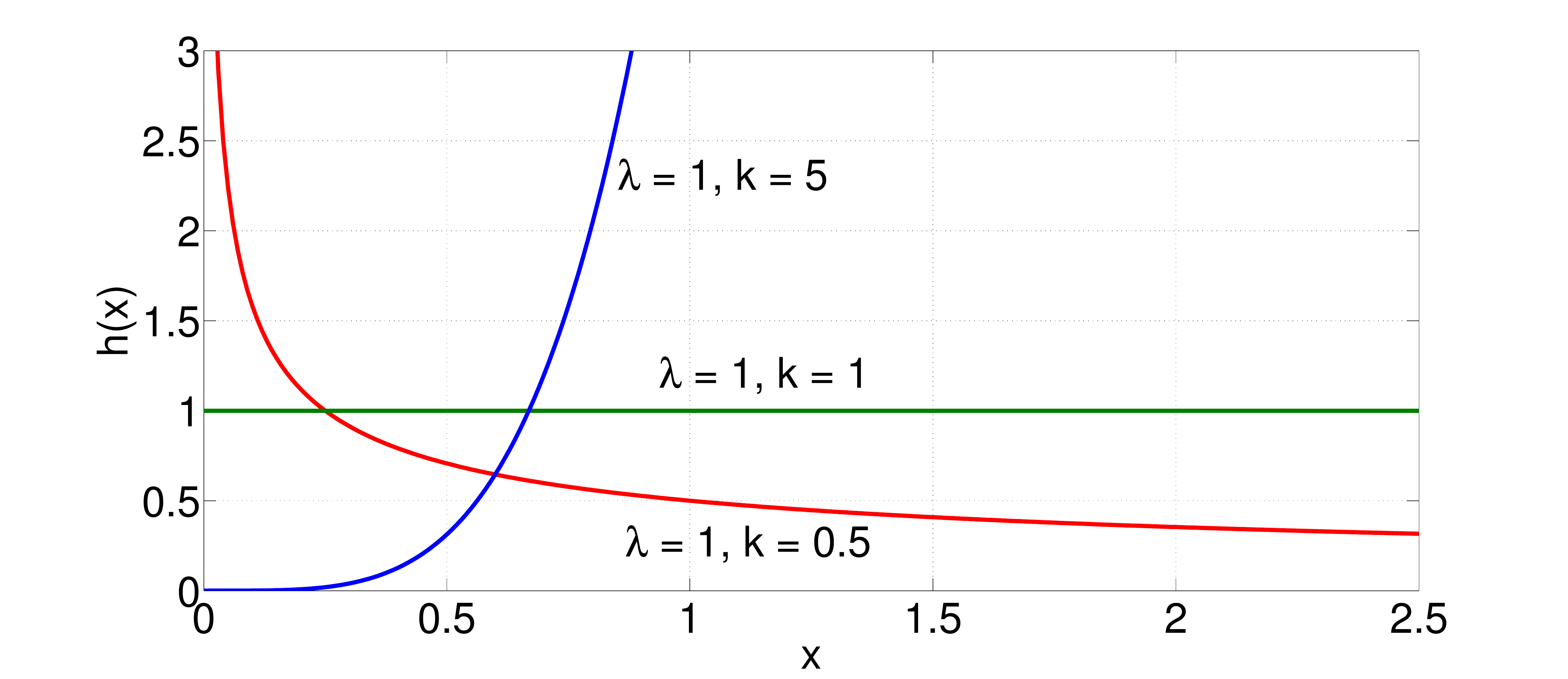} \label{fig:weibull_h}
  } 
  \caption{Weibull distribution for different values of $k$} 
\end{figure}

The HILT o.d.e. under Weibull distribution of threshold can be written as follows:
\[\dot{b} = d\]
\[\dot{d} = -d + \Gamma d \frac{k}{\lambda} \bigg(\frac{\Gamma b}{\lambda}\bigg)^{k-1}(1-b-d)\]

Figures~\ref{fig:weibull_evol} and \ref{fig:weibull_evol_large_gamma} demonstrate the evolution of the o.d.e under the Weibull distribution of threshold, for different values of $k$ in the small and large regimes for $\Gamma$. For smaller $\Gamma$ (Figure~\ref{fig:weibull_evol}), we observe that as $k$ increases, the spread of influence decreases. This is expected, since from Figure~\ref{fig:weibull_f} it is clear that, for larger values of $k$, Weibull distribution puts more mass on larger values of threshold, i.e., nodes are less susceptible to influence. Further, for $k=5$, Figure~\ref{fig:weibull_evol} shows that the total spread of influence is $0.2$, equal to the initial seeding $d_0 = 0.2$. This implies the influence does not spread at all, since the node thresholds are much higher, compared to the net influence generated by $d_0$ (due to smaller $\Gamma$).

For larger $\Gamma$ (Figure~\ref{fig:weibull_evol_large_gamma}), we see that the trend is reversed, i.e., as $k$ increases, the spread of influence increases. It is to be noted that the $\dot{b} = d$ near $0$ is larger for smaller $k$, similar to the small $\Gamma$ regime. However, from Figures~\ref{fig:weibull_f} and \ref{fig:weibull_h} we see that smaller values of $k$ have heavier tails (and lower hazard rates), thus leading to stagnation of influence after the initial surge. 

Another interesting feature to note is that, unlike the small $\Gamma$ regime, for $k=5$, we get a much higher influence spread. Also, unlike other values of $k$, here $\dot{b} = d$ exhibits a non-monotonic behavior even after it begins to decrease, i.e., $d(t)$ is not unimodal. Such behavior has not been observed until now in the classic epidemiology framework, especially in a homogeneous setting. In traditional epidemic models like \emph{SIR}, the \emph{I} process (equivalent to $\dot{b}$) might exhibit an initial increase, but once it begins decreasing, continues to steadily decrease to zero. But, in our dynamics, the presence of hazard rate (increasing, in this case) leads to such non-unimodal characteristics of $d(t)$.

\begin{figure}[ht!] 
  \subfigure[small $\Gamma$ regime]{ 
    \includegraphics[scale=0.17]{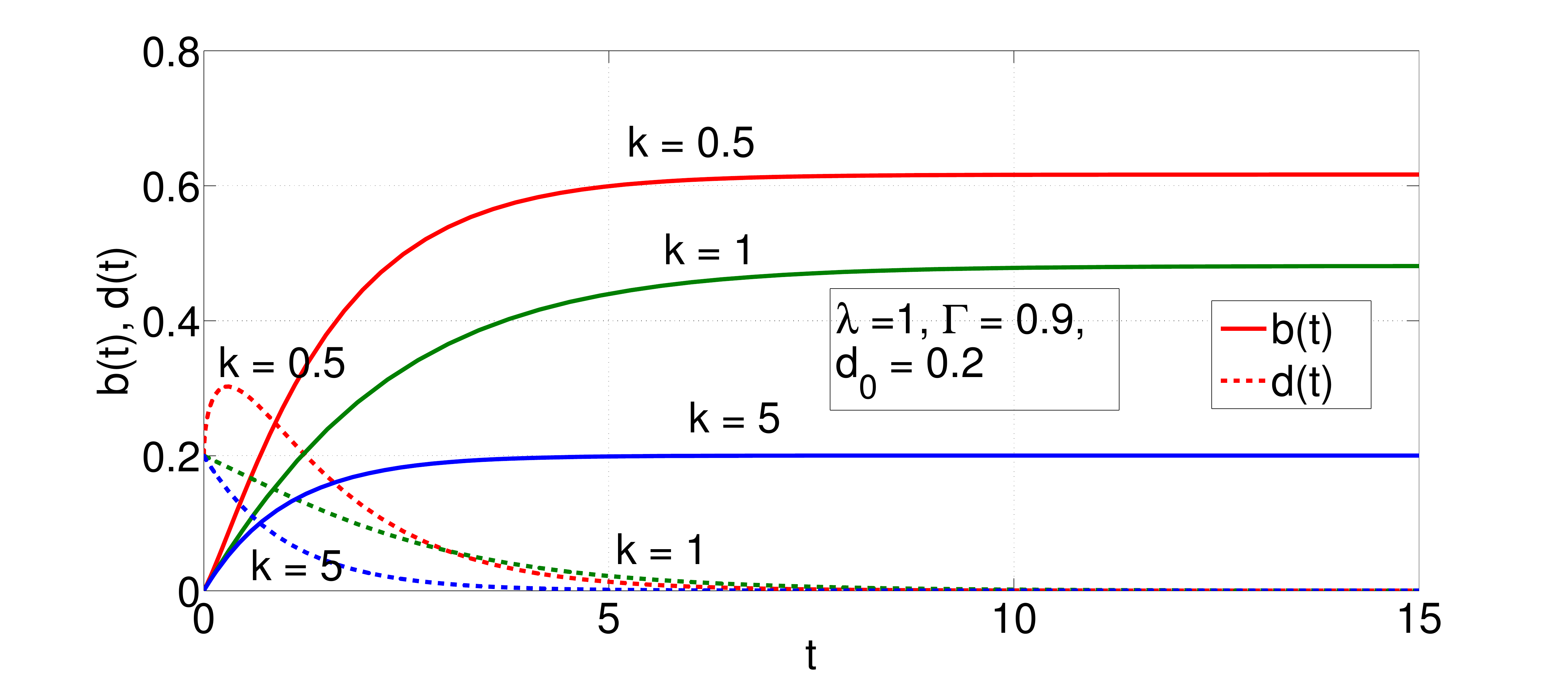} \label{fig:weibull_evol}
  } 
  \quad 
  \subfigure[large $\Gamma$ regime]{
    \includegraphics[scale=0.17]{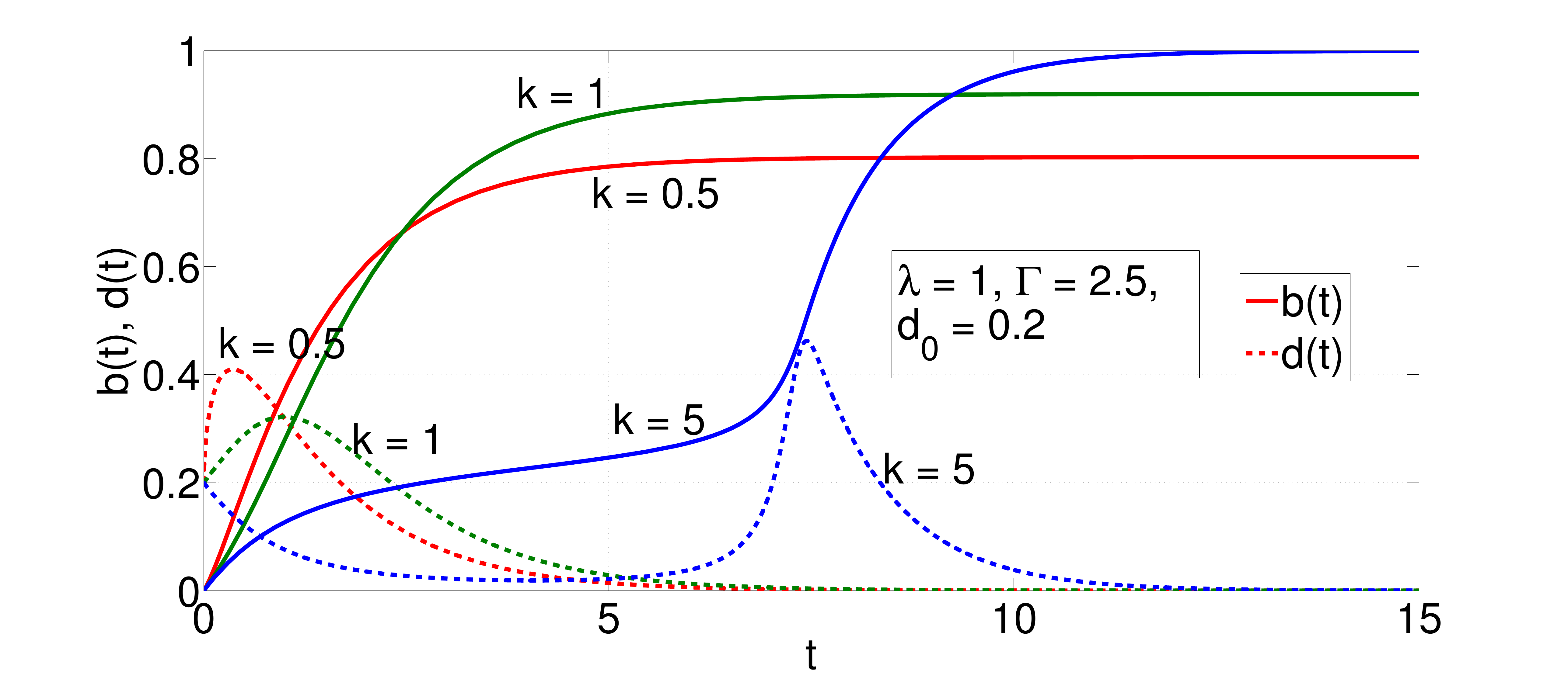} \label{fig:weibull_evol_large_gamma}
  } 
  \caption{Comparison of influence spread between Weibull threshold distributions with different values of $k$} 
\end{figure}

\subsection{Loglogistic distribution}
Loglogistic distribution is the probability distribution of a random variable who logarithm follows the logistic distribution. It has similar shape characteristics to log-normal distribution, but has heavier tails. The probability density function and the hazard function are given by, 

\[ f(x; \alpha, \beta) = \frac{(\frac{\beta}{\alpha})(\frac{x}{\alpha})^{\beta-1}}{(1 + (\frac{x}{\alpha})^\beta)^2}, \ x \geq 0 \]

\[ h(x; \alpha, \beta) = \frac{\beta}{\alpha} \bigg[ \frac{(\frac{x}{\alpha})^{\beta-1}}{1 + (\frac{x}{\alpha})^\beta} \bigg] , \ x \geq 0 \]

The parameter $\alpha$ functions as the scale parameter and $\beta$ is referred to as the shape parameter. Also for $\beta > 1$, the distribution is unimodal, and is more concentrated as $\beta$ increases (see Figure~\ref{fig:loglogistic_f}).

Similar to the Weibull distribution, one can obtain different failure characteristics by tuning the $\beta$ parameter. For $\beta \leq 1$, the hazard rate decreases monotonically. But unlike the Weibull distribution, for $\beta >1$, the hazard function exhibits non-monotonic behavior (see Figure~\ref{fig:loglogistic_h}). 

\begin{figure}[ht!] 
  \subfigure[Probability density function]{ 
    \includegraphics[scale=0.17]{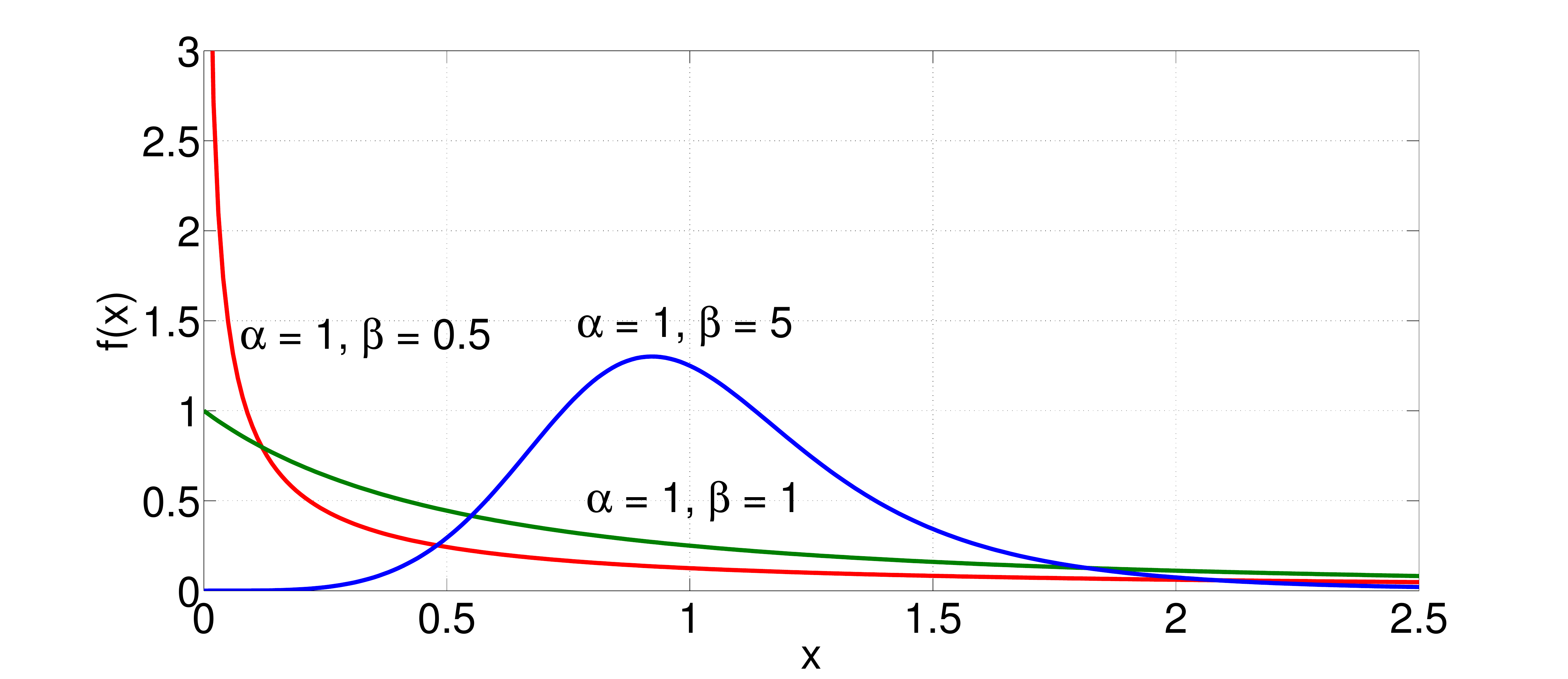} \label{fig:loglogistic_f}
  } 
  \quad 
  \subfigure[Hazard function]{
    \includegraphics[scale=0.17]{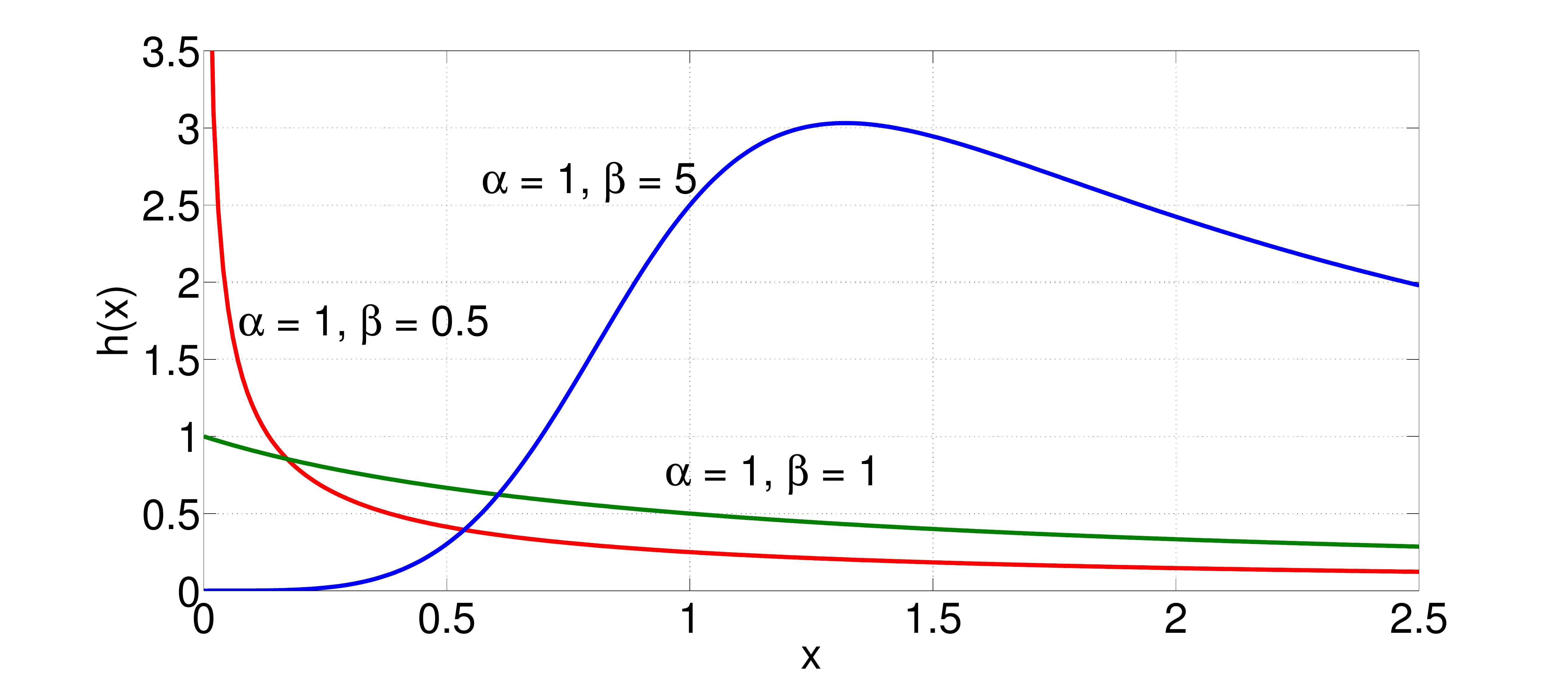}\label{fig:loglogistic_h}
  } 
  \caption{loglogistic distribution for different values of $k$} 
\end{figure}

\pagebreak
The HILT o.d.e. under loglogistic distribution of threshold can be written as follows:
\[\dot{b} = d\]
\[\dot{d} = -d + \Gamma d \frac{\beta}{\alpha} \bigg[ \frac{(\frac{\Gamma b}{\alpha})^{\beta-1}}{1 + (\frac{\Gamma b}{\alpha})^\beta} \bigg] (1-b-d)\]

Figures~\ref{fig:loglogistic_evol} and \ref{fig:loglogistic_evol_large_gamma} demonstrate the evolution of the influence spread o.d.e. under the loglogistic distribution of threshold, for different values of $k$ in the small and large regimes for $\Gamma$. We note that for small $\Gamma$, the evolution of influence is qualitatively similar, but under the loglogistic distribution, we get a smaller influence spread, due to heavier tails. Also, in the large $\Gamma$ regime, we note that for $\beta=5$, we again get a non-unimodal behavior for $d(t)$. But the second peak is less pronounced in the loglogistic distribution than the Weibull distribution, since the loglogistic distribution exhibits a non-monotonic hazard rate.

\begin{figure}[ht!] 
  \subfigure[small $\Gamma$ regime]{ 
    \includegraphics[scale=0.17]{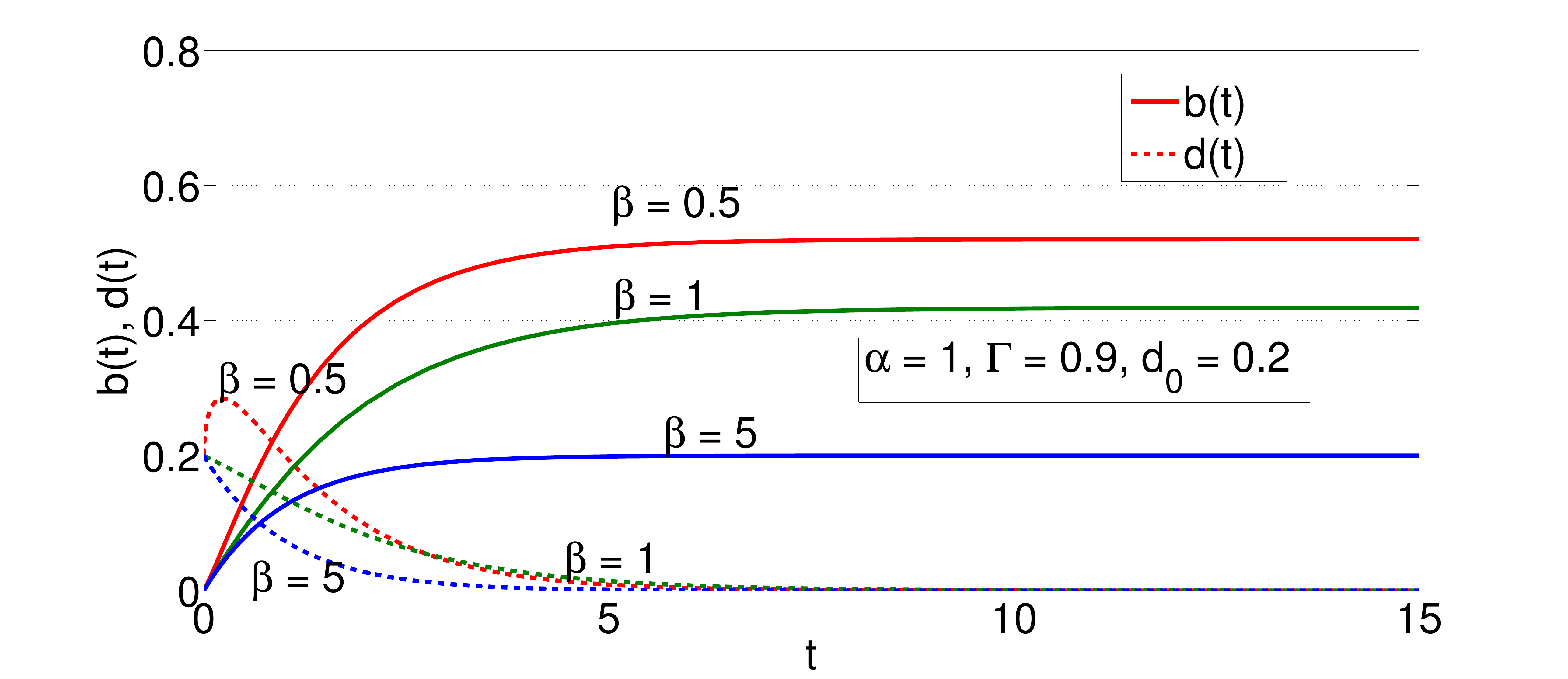} \label{fig:loglogistic_evol}
  } 
  \quad 
  \subfigure[large $\Gamma$ regime]{
    \includegraphics[scale=0.17]{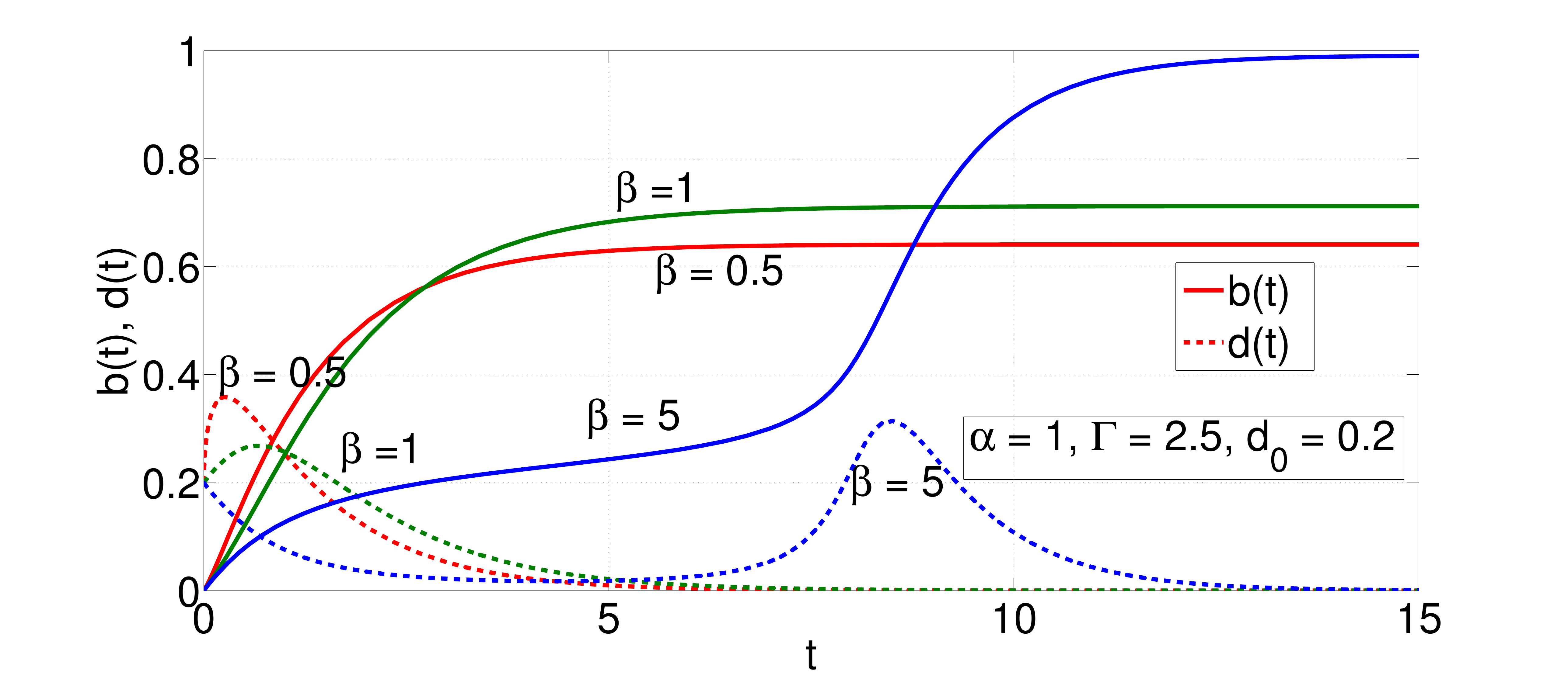}\label{fig:loglogistic_evol_large_gamma}
  } 
  \caption{Comparison of influence spread between loglogistic threshold distributions with different values of $k$} 
\end{figure}

%
%
%

Thus we see that, the incorporation of hazard rate into the o.d.e. (resulting from a fluid limit characterization of the LT model) yields qualitatively different characteristics compared to the standard epidemic models. To the best of our knowledge, this is the first work that analytically characterizes the evolution of influence under different threshold distributions. This is also the first work to incorporate hazard functions into the epidemic models, thus providing a way to capture heterogeneity in the population. Further, due to the one-one correspondence between a given hazard function and its corresponding cumulative distribution \cite{cox61renewal-theory}, one can begin with the hazard function in the o.d.e. (obtained by curve-fitting to existing epidemic data) and ascertain the threshold distribution of the population.  

\section{Multiclass HILT model}
\label{sec:multiclass}
A natural extension to the HILT model would be to consider the evolution of information spread in an heterogeneous network. Such a scenario might arise in a network with communities, where the interactions within a community might be stronger than the interaction across communities. These have been traditionally studied under the term \emph{stratified epidemics} \cite{watson72epidemic}. Consider a network with $M$ communities $(\mathcal{C}_i)_{i=1}^{M}$ and let $(N_i)_{i=1}^{M}$ denote the number of nodes in each community. Let $\mathcal{G}$ be the influence matrix, whose entries $g_{i,j}$ indicates the strength of influence from community $i$ to community $j$ (see Figure~\ref{fig:multiclass}).

\begin{figure}
\centerline{\includegraphics[scale=0.4]{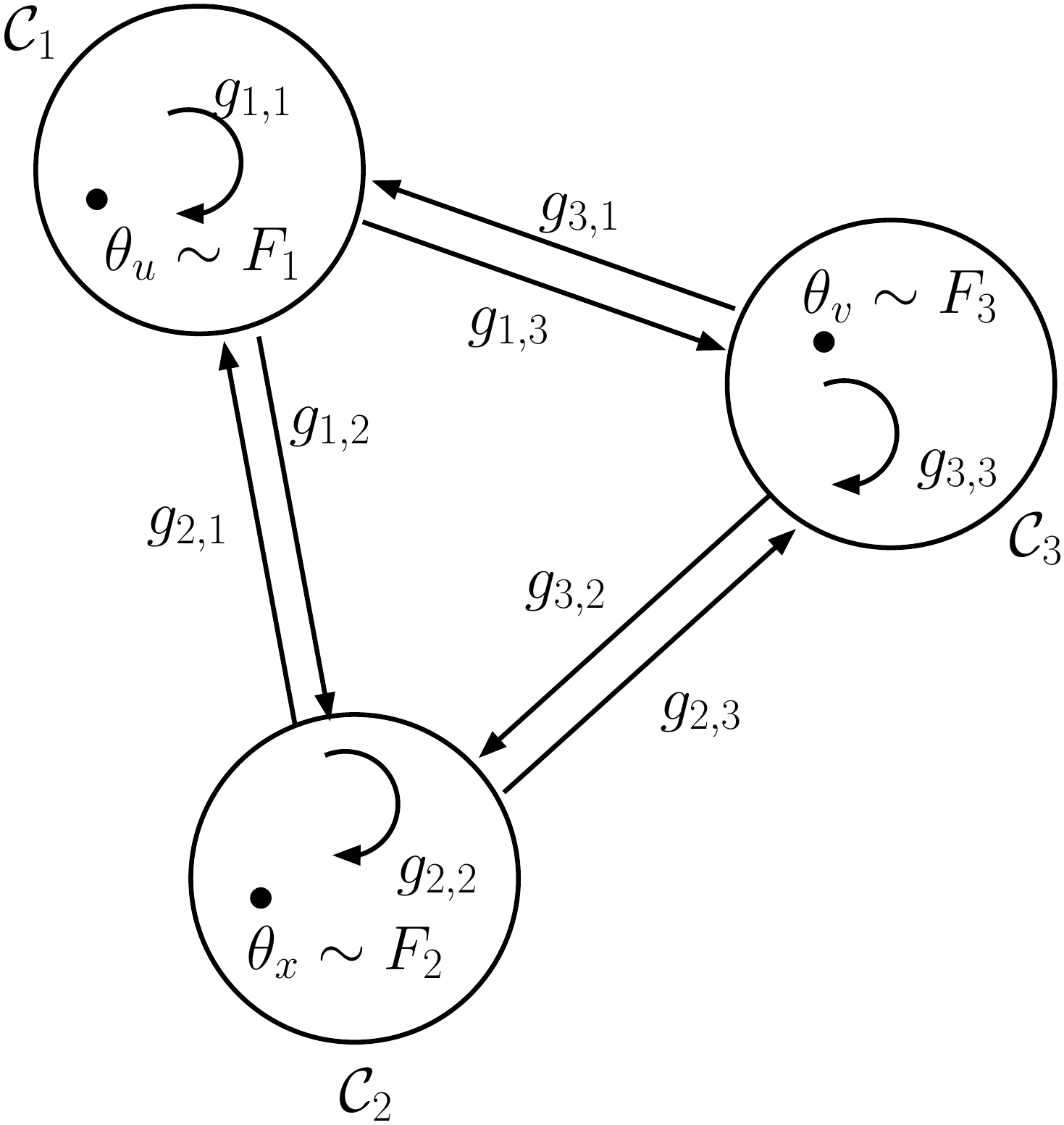}}
\caption{A heterogeneous network with three communities, shown with entries of the influence matrix $\mathcal{G}$. Nodes $u$, $x$ and $v$ belong to communities $\mathcal{C}_1$, $\mathcal{C}_2$ and $\mathcal{C}_3$, and have their thresholds distributed according to $F_1(\cdot)$, $F_2(\cdot)$ and $F_3(\cdot)$ respectively.}
\label{fig:multiclass}
\end{figure}

As earlier, we will appropriately normalize the edge weights, i.e., for $u \in \mathcal{C}_i$ and $v \in \mathcal{C}_j$, $w_{u,v} = \frac{g_{i,j}}{N}$, where $N = \sum_{i} N_i$ is the total population size. Let the nodes within community $i$ have their thresholds distributed according to $F_i$, with hazard function $h_{F_i}$. We can then carry out an analysis similar to what was done for the HILT model in Section~\ref{sec:ode-approx}. We can show that the joint evolution is a Markov process, and we construct a scaled process using the minislots approach and with appropriate probability scaling. Here the attempt probability of all infectious nodes during a given mini-slot scales as $\frac{1}{N}$, irrespective of which community they belong to. By applying Kurtz's theorem to the scaled process, we obtain the o.d.e.s representing the influence evolution. Let $(b_i(t), d_i(t))$ denote the non-infectious and infectious active nodes within community $i$. We can then describe their evolution by the following system 
of o.d.e.s similar to Equations~\ref{eqn:b_limit} and \ref{eqn:d_limit} ($1 \leq i \leq M$):
\[ \dot{b}_i = d_i\]
\[ \dot{d}_i = -d_i + [ \mathcal{G}^T \mathbf{d} ]_i \ h_{F_i} ( [\mathcal{G}^T \mathbf{b}]_i ) (n_i - b_i - d_i)\]
where $\mathbf{b} = (b_1, b_2, \cdots b_m)^T$, $\mathbf{d} = (d_1, d_2, \cdots d_m)^T$ and $n_i = \lim_{N \to \infty} \frac{N_i}{N}$.
One possible objective function to maximize in this scenario would be the total spread of influence $\sum_{i} b_i (\infty) $, by suitably choosing the initial $\mathbf{d}(0)$ subject to the constraint $\sum_i d_i(0) = d_0$, for fixed system parameters, i.e., the threshold distributions $F_i$ and the influence matrix $\mathcal{G}$. We were unable to obtain a universal analytical solution for this problem, but numerically demonstrate that depending the system parameters the results could be quite counter-intuitive. 
\begin{figure}
\centerline{\includegraphics[scale=0.17]{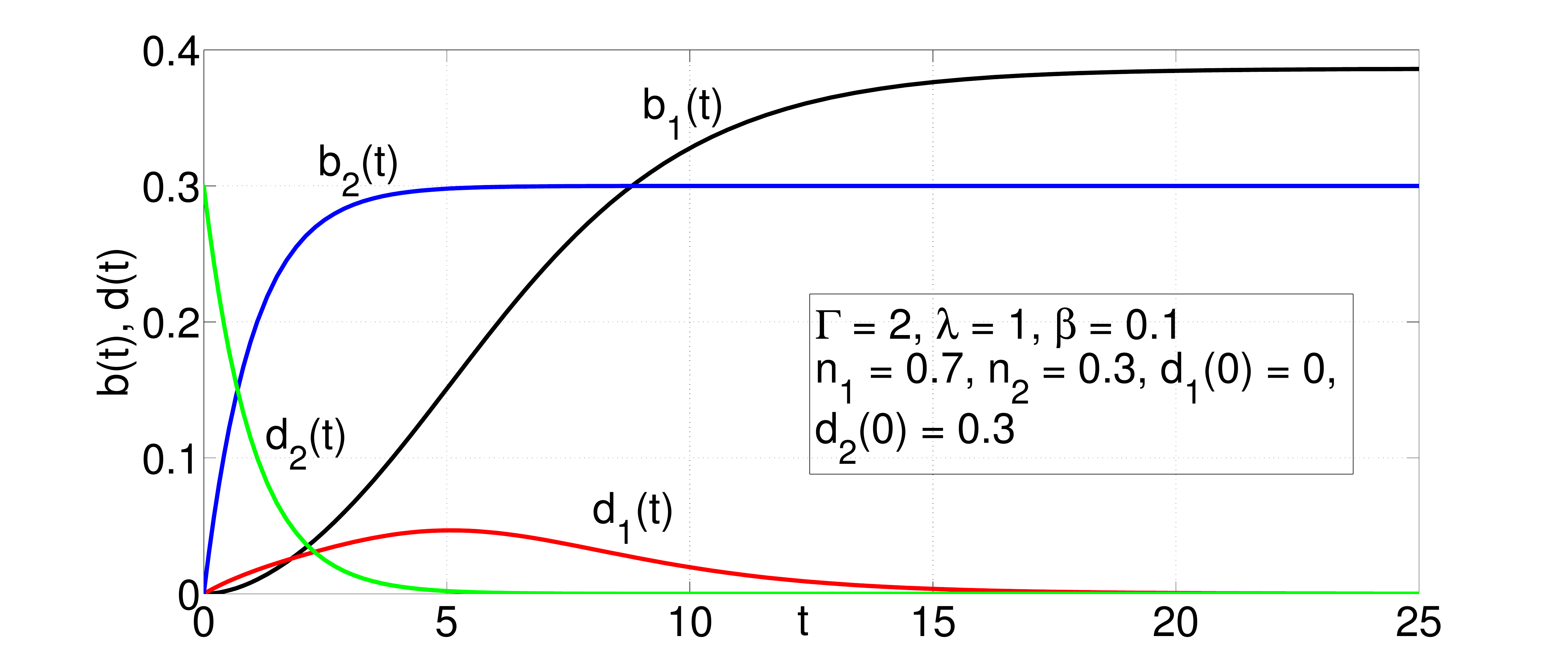}}
\caption{Evolution of influence in the two communities when the initial seeding is done in the smaller community  (i.e., $d_1(0) = 0, d_2(0) = d_0 = 0.3$ )}
\label{fig:multiclass_invest_smaller}
\end{figure}

\begin{figure}
\centerline{\includegraphics[scale=0.17]{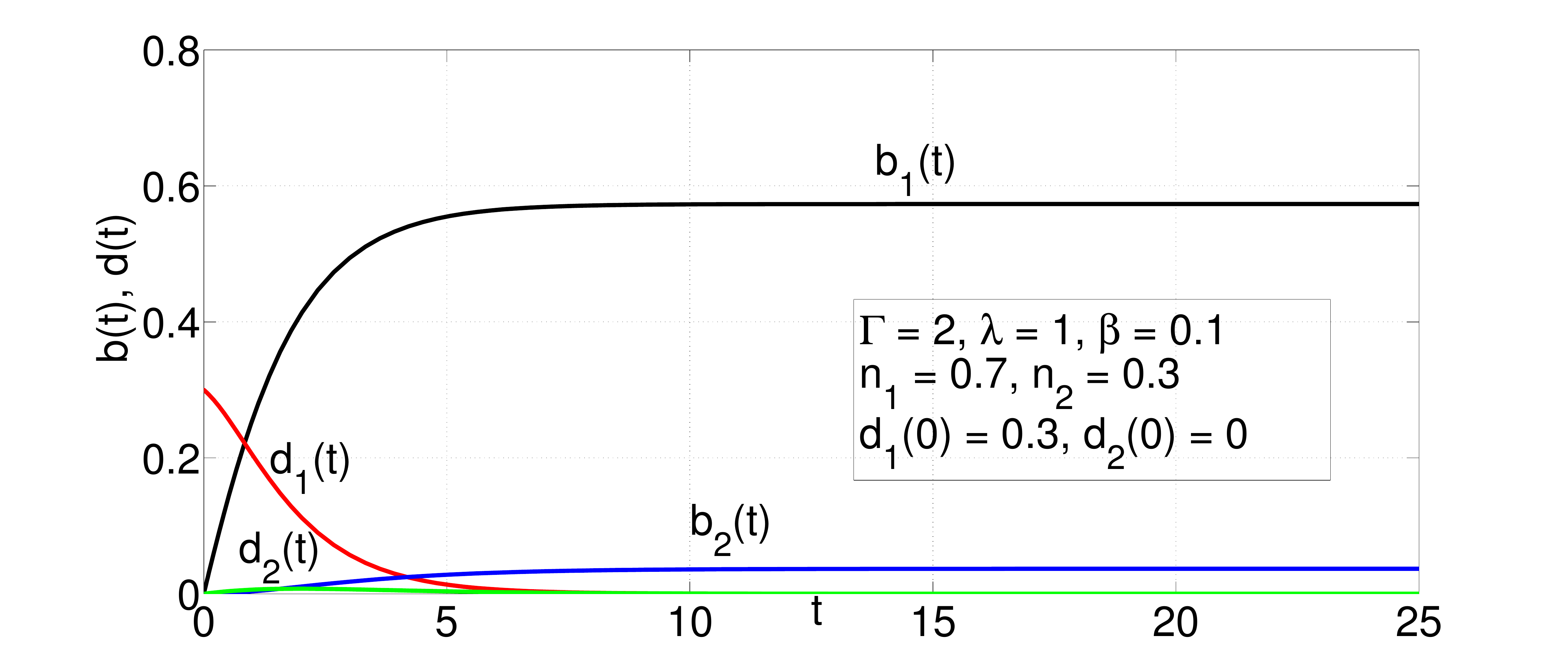}}
\caption{Evolution of influence in the two communities when the initial seeding is done in the larger community  (i.e., $d_1(0) = d_0 = 0.3, d_2(0) = 0$ )}
\label{fig:multiclass_invest_larger}
\end{figure}

For instance, consider a two community network with $N_1 = 0.7 N$ and $N_2 = 0.3 N$ as the relative community sizes. Let all the nodes in the population have their thresholds distributed according to an exponential distribution with parameter $\lambda$. Also assume that $g_{i,i} = \Gamma$ and $g_{i,j} = \beta$ for $i,j \in \{1,2\}$. Figures \ref{fig:multiclass_invest_smaller} and \ref{fig:multiclass_invest_larger} show the evolution of $(b_i(t), d_i(t))_{i}$ for $i = 1,2$, for different initial conditions. While in Figure~\ref{fig:multiclass_invest_smaller} (scenario 1) the entire initial seeding is done in the smaller community (i.e., $d_1(0) = 0, d_2(0) = d_0 = 0.3$ ), in Figure~\ref{fig:multiclass_invest_larger} (scenario 2) the entire initial seeding is done in the larger community. We see that, the total spread of influence in scenario 1 is larger than in scenario 2. Further from Figure \ref{fig:twoclass_n1_variation} it is clear that the optimal seeding for this setting is approximately $(d_1(0) = 0.1, 
d_2(0) = 0.2)$. It is surprising that we get a wider spread of influence by investing more in the smaller community. Thus we see that, even in a simple two community setting, the optimal seeding might be counter-intuitive. It would be an interesting future direction to analytically obtain the optimal seeding, given the influence matrix $\mathcal{G}$ and the threshold distributions $\{F_i\}_{1\leq i \leq m}$. 

\begin{figure}
\centerline{\includegraphics[scale=0.17]{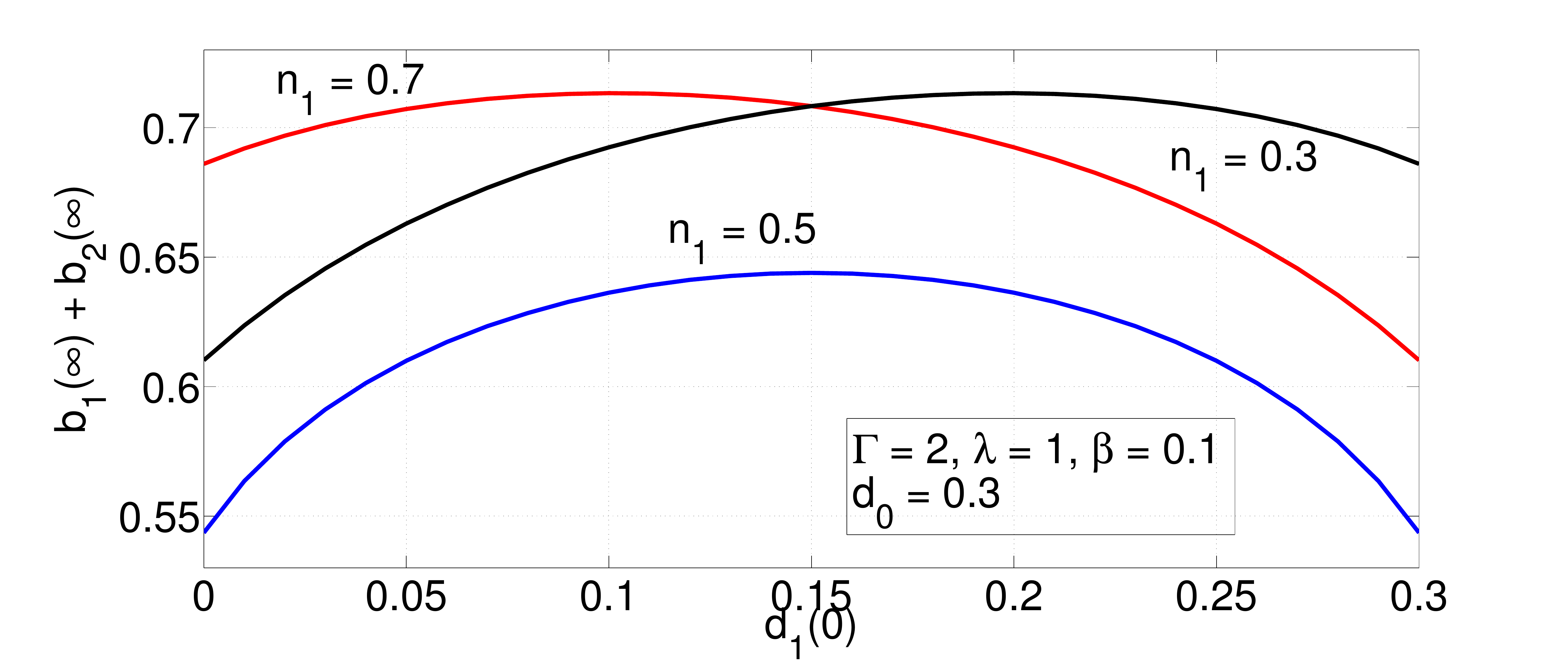}}
\caption{Total spread of influence $b_1(\infty) + b_2(\infty)$ for various allocations of initial seeding $(d_1(0), d_2(0))$ and different relative community sizes.}
\label{fig:twoclass_n1_variation}
\end{figure}

\section{Conclusion}
\label{sec:conclusion}
In this work, we began with a homogeneous version of the Linear Threshold model proposed by Kempe et al. \cite{kempe-etal03max-spread-infl} in the context of viral marketing, and generalized it for arbitrary threshold distributions. We observed that the spread of influence evolves as a discrete time Markov chain. Under a certain scaling, we showed that the scaled Markov chain converges (in the sense of \cite{kurtz70ode-markov-jump-processes}) to a deterministic trajectory defined by an o.d.e.. The threshold distribution appears in terms of its hazard rate function in this o.d.e. We described how this approach complements the fixed point equation suggested by Granovetter \cite{granovetter78threshold-models}, thus providing a link between two threads in the threshold model literature. Also, under the exponential distribution of threshold, we showed that the derived fluid dynamics are equivalent to the well-known SIR model in epidemiology. We also numerically demonstrated how incorporating the hazard function 
into the o.d.e. can provide qualitatively different characteristics compared to traditional epidemic models, even in a homogeneous setting. One of the interesting future directions is to incorporate the degree distribution of the underlying network in the fluid dynamics. Further, one can carry out a similar analysis for influence processes with a general threshold function (instead of linear), as indicated in \cite{kempe-etal03max-spread-infl}. Also, using the available social network data and via controlled experiments, one could validate or suggest improvements to the threshold model, in order to fit the real world dynamics. 

\bibliographystyle{unsrt}
\bibliography{tnse_hilt}

\begin{appendices}
\section{DTMC $(B(k), D(k))$}
\label{app:dtmc}
Let $\mathcal{F}_k$ denote the entire history of the processes up to time $k$, i.e., $\mathcal{F}_k = ( B(l),D(l) )_{l=0}^{l=k}$. To obtain the expected drift of $D(k)$, consider, 

\begin{eqnarray*}
\lefteqn{\mathbb{P} (D_{k+1} = l | \mathcal{F}_k )}\\
&=& \mathbb{P} (\sum_{j \notin B_{k}} I \{ b_j(B_{k}) < \theta_j \leq b_j(B_{k+1}) \} = l | \mathcal{F}_k )\\
&=& \hspace{-0.8cm} \sum_{\stackrel{L \subseteq N\backslash B_{k+1}}{|L| =l}} \prod_{j\in L}  \mathbb{P} (I \{ b_j(B_{k}) < \theta_j \leq b_j(B_{k+1}) \} = 1 | \mathcal{F}_k )  \times \\
& & \hspace{1cm} \prod_{j\notin L} \mathbb{P} (I \{ b_j(B_{k}) < \theta_j \leq b_j(B_{k+1}) \} = 0 | \mathcal{F}_k )\\
\end{eqnarray*}

Let $|B_k| = b$ and $|D_k| =d$, then we have in the HILT model $b_j(B_{k}) = \Gamma b$ and $b_j(B_{k+1}) = \Gamma (b+d)$. Hence we can write, 

\begin{eqnarray*}
\lefteqn{\mathbb{P} (D_{k+1} = l | \mathcal{F}_k )}\\
&=& \sum_{L \subseteq N\backslash B_{k+1},|L| =l} \prod_{j\in L}  \mathbb{P} (\theta_j \leq \Gamma (b+d) | \theta_j > \Gamma b ) \times \\
& & \hspace{2cm} \prod_{j\notin L} \mathbb{P} (\theta_j > \Gamma (b+d) | \theta_j > \Gamma b )\\
&=& \sum_{L \subseteq N\backslash B_{k+1}|L| =l} \bigg( \frac{F(\Gamma (b+d)) - F(\Gamma b)}{1-F(\Gamma b)}\bigg)^{l} \times \\
& & \hspace{1.8cm}\bigg(1- \frac{F(\Gamma (b+d)) - F(\Gamma b)}{1-F(\Gamma b)}\bigg)^{(N-b-d-l)}\\
&=& \binom{N-b-d}{l} \bigg( \frac{ F(\Gamma (b+d)) - F(\Gamma b)}{1-F(\Gamma b)} \bigg) ^l \times \\
& & \hspace{1.8cm} \bigg(1 - \frac{ F(\Gamma (b+d)) - F(\Gamma b)}{1-F(\Gamma b)} \bigg) ^{(N-b-d-l)}\\
&=& P(D(k+1) = l | B(k)= b, D(k) = d)\\
\end{eqnarray*}

From the above equations we can clearly see that $(B(k), D(k))$ is a DTMC on the state space $[0,1,\cdots ,N]\times[0,1, \cdots, N-B(k)]$.

\section{Scaling the HILT Model}
\label{app:scaling-kurtz}
In this section, we will demonstrate the necessity for a probabilistic scaling (in addition to the amplitude and time scaling) to arrive at the mean drift expressions. Let $\mathcal{F}_k$ denote the entire history of the processes up to time $k$, i.e., $\mathcal{F}_k = ( B(l),D(l) )_{l=0}^{l=k}$. Begin with the drift equations for the unscaled process $(B(k),D(k))$. 
\[ \mathbb{E} \bigg[ B(k+1) - B(k)| \mathcal{F}_k \bigg] = D(k) \]
\begin{eqnarray*}
\lefteqn{E \bigg[ D(k+1) - D(k)| \mathcal{F}_k \bigg]}\\
&=& - D(k)+ \\
& & \frac{F(\gamma(B(k)+D(k))) - F(\gamma B(k))}{1-F(\gamma B(k))} \times \\
& & \hspace{2cm}(N- D(k)- B(k))
\end{eqnarray*}
We shall now try scaling the process in the usual way, i.e.,by scaling down the amplitude by a factor of $N$, $\Tilde{B}^N(k) = \frac{B(k)}{N}$, $\Tilde{D}^N(k) = \frac{D(k)}{N}$. The evolution equations can then be written down as follows:
\[ E \bigg[ \Tilde{B}^N(k+1) - \Tilde{B}^N(k)| \mathcal{F}_{k} \bigg] = \Tilde{D}^N(k) \]
\begin{eqnarray*}
\lefteqn{E \bigg[ \Tilde{D}^N(k+1) - \Tilde{D}^N(k)| \mathcal{F}_{k} \bigg]}\\ 
&=& - \Tilde{D}^N(k)+ \\ 
& & \frac{F(\gamma N(\Tilde{B}^(k)+\Tilde{D}^N(k))) - F(\gamma N \Tilde{B}^N(k))}{1-F(\gamma N \Tilde{B}^N(k))} \times \\
& & \hspace{2cm}(1- \Tilde{D}^N(k) - \Tilde{B}^N(k))
\end{eqnarray*}
Using $d= \Tilde{D}^N(k)$, $b=\Tilde{B}^N(k)$ and $\Gamma = \gamma N$, we can write the drift function as, 
\[ f_N  = \bigg( d, \frac{F(\Gamma (b+d) - F(\Gamma b)}{1-F(\Gamma b)} (1-b-d) - d \bigg) \]
where both $d$ and $b$ are fractions taking values from $[0,1]$. It is clear that $\frac{f_N}{1/N}$ diverges with $N \rightarrow \infty$ but we want this quantity to converge to a function $f$ (which is independent of $N$) so that we can apply Kurtz's theorem to obtain an approximating ODE. We can see that the problem in the above case is caused because the drift function in the original process scales with the state. Hence in this case, while scaling, we need to slow down the process by another factor of $N$. To this purpose, we use the probabilistic attempt model in our scaling. The same scaling has been used in the literature in the context of the analysis of Random Multi-Access Algorithms by Bordenave et al.~\cite{bordenave05mean-fields}. Note that this modifies the dynamics of the original process. The o.d.e. will be the limit (in probability) of the stochastic process with modified dynamics as $N \rightarrow \infty$ but will be a heuristic approximation for the stochastic process with the original 
dynamics. 

\section{Proof of Theorem 1}
\label{app:theorem1}

Kurtz's theorem \cite{kurtz70ode-markov-jump-processes} provides us a way by which we can approximate the evolution of a pure jump Markov process by the solution of a derived ODE. In this paper we shall refer to \cite{darling02limits-purejump-markov} for an equivalent version of Kurtz's theorem, which is simpler to handle. It can be restated as follows to be directly used in our context.

\begin{theorem}
Given that,
\begin{itemize}
\item[(i)] $f(b,d)$ is Lipschitz

\item[(ii)] $\sup_{(b,d) \in \Delta^{(N)}} \bigg{|} \frac{f^{(N)}(b,d) }{\frac{1}{N}} - f(b,d) \bigg{|} \stackrel{N\rightarrow \infty}{\rightarrow} 0$ where $\Delta^{(N)} = [0,\frac{1}{N},\cdots ,N]\times[0,\frac{1}{N}, \cdots,1]$, with $b+d \leq 1$.

\item [(iii)] $E(| \Tilde{Y}^{(N)}(k)|^2 / \mathcal{F}_{k}) \leq \frac{C_0}{N^2}$ and  $E(| \Tilde{Z}^{(N)}(k)|^2 / \mathcal{F}_{k}) \leq \frac{C_1}{N^2}$ where \\ $\mathcal{F}_{k} = (B^{(N)}(0), D^{(N)}(0), \cdots, B^{(N)}(k), D^{(N)}(k))$ is the history of the process upto time $k$.

\item [(iv)] $\Tilde{B}^{(N)}(0) \stackrel{p}{\rightarrow} b(0)$ and $\Tilde{D}^{(N)}(k) \stackrel{p}{\rightarrow} d(0)$

\end{itemize}

then we have for each $T > 0$ and each $\epsilon > 0$,

\begin{eqnarray}
 P \bigg( \sup_{0 \leq t \leq T} \big| \big| \big( \Tilde{B}^{N}( \lfloor Nt \rfloor ),\Tilde{D}^{N}( \lfloor Nt \rfloor ) \big) - \big( b(t),d(t) \big) \big| \big| > \epsilon \bigg) &&  \nonumber \\
 & & \hspace{-2cm} \stackrel{N\rightarrow \infty}{\rightarrow} 0 \nonumber 
\end{eqnarray}

where $b(t)$ and $d(t)$ are defined as the solutions of the system of ODE,

\[\dot{b}(t) = f_1(b,d) \]
\[\dot{d}(t) = f_2(b,d) \]

with initial conditions $(b(0),d(0))$.

\end{theorem}

\begin{itemize}
\item[(i)]\textbf{Lipschitz property}\\
Consider,
\[f_1(b,d)= d\]
\[ f_2(b,d) = \frac{ \Gamma d}{1 - \Gamma b} (1 - b - d) - d\]

\[ \frac{\partial f_1}{\partial b} = 0 ;  \frac{\partial f_1}{\partial d} =  1\]

\[ \frac{\partial f_2}{\partial b} = \frac{\Gamma (d( \Gamma - 1) - d^2 \Gamma)}{(1-\Gamma b)^2}; \frac{\partial f_2}{\partial d} = \frac{ \Gamma -1}{1 - \Gamma b} - \frac{2d \Gamma}{1 - \Gamma b}\]

We see that each of the terms above is bounded when $(b,d) \in [0,1]\times[0,1-b]$. Thus the norm of Jacobian $|| Df(b,d) ||$ is uniformly bounded, and it follows that $f(b,d) = (f_1(b,d),f_2(b,d))$ is Lipschitz.

\item[(ii)]\textbf{Uniform Convergence}\\
\begin{eqnarray*}
f^{N}(b,d) = \frac{1}{N} \bigg(d, \frac{N \gamma d}{1-N \gamma b} (1-b-d) - d \bigg) 
\end{eqnarray*}

\begin{eqnarray*}
f(b,d) = \bigg(d,\frac{\Gamma d}{1-\Gamma b} (1-b- d) - d \bigg)
\end{eqnarray*}

By definition, $\Gamma = N \gamma$ and hence the uniform convergence of $\frac{f^{N}(b,d)}{1/N}$ to $f(b,d)$ in the domain $(b,d) \in [0,\frac{1}{N},\cdots ,N]\times[0,\frac{1}{N}, \cdots,b]$ is straightforward. 

\item[(iii)]\textbf{Bounded Noise variance}\\
We can write the noise variances as follows:
\begin{eqnarray*}
E(| Y^{N}(k)|^2 | \mathcal{F}_{k}) &=& \frac{1}{N}(1-\frac{1}{N})D^{N}(k) \\
&\leq & \Tilde{D}^{N}(k) \leq 1\\
\end{eqnarray*}

\begin{eqnarray*}
\lefteqn{E(| Z^{N}(k)|^2 | \mathcal{F}_{k})=}\\
& & \hspace{-1cm} \frac{1}{N}(1-\frac{1}{N})D^{N}(k) + \sum_{D^{\star N}(k)=0}^{D^{N}(k)} \bigg( \frac{{D}^{\star N}(k) \gamma}{1-\gamma B^{N}(k)} \\
& & \hspace{-1cm} \times \big( 1 - \frac{{D}^{\star N}(k) \gamma}{1-\gamma B^{N}(k)} \big) (N - B^{N}(k) - D^{N}(k)) \\
& & \hspace{-1cm} \times \binom{D^{N}(k)}{{D}^{\star N}(k)} \big( \frac{1}{N} \big)^{{D}^{\star N}(k)} \big(1- \frac{1}{N} \big)^{D^{N}(k)-{D}^{\star N}(k)} \bigg) \\
& \leq & \Tilde{D}^{N}(k) + \frac{\Gamma (1-\Tilde{B}^{N}(k)-\Tilde{D}^{N}(k))\Tilde{D}^{N}(k)}{1-\Gamma \Tilde{B}^{N}(k)}\\
\end{eqnarray*}

where ${D}^{\star N}(k)$ represents the number of nodes that succeed in contributing their influence, at the mini slot $k$. Since $\Tilde{D}^{N}(k)$,$\Tilde{B}^{N}(k)$ and $\Gamma$ are less than or equal to 1, both the above terms can be bounded above by constants $C_1$ and $C_2$. Hence in the process involving fraction of nodes, we have

\begin{eqnarray*}
E(| \Tilde{Y}^{N}(k)|^2 | \mathcal{F}_{k}) &=& \frac{1}{N^2} E(| Y^{N}(k)|^2 | \mathcal{F}_{k})\\
&\leq &  \frac{C_1}{N^2}\\
\end{eqnarray*}

\begin{eqnarray*}
E(| \Tilde{Z}^{N}(k)|^2 | \mathcal{F}_{k}) &=& \frac{1}{N^2} E(| Z^{N}(k)|^2 | \mathcal{F}_{k})\\
&\leq &  \frac{C_2}{N^2}\\
\end{eqnarray*}
and we can see that the noise conditions are satisfied.

\item[(iv)]\textbf{Convergence of initial conditions}\\
By choice, we have $\Tilde{B}^{N}(0) = b(0)$ and $\Tilde{D}^{N}(0)=d(0)$. 

\end{itemize}

Thus by Kurtz's theorem, we have for each $T > 0$ and each $\epsilon > 0$,

\begin{eqnarray}
 P \bigg( \sup_{0 \leq t \leq T} \big| \big| \big( \Tilde{B}^{N}( \lfloor Nt \rfloor ),\Tilde{D}^{N}( \lfloor Nt \rfloor ) \big) - \big( b(t),d(t) \big) \big| \big| > \epsilon \bigg) &&  \nonumber \\
 & & \hspace{-2cm} \stackrel{N\rightarrow \infty}{\rightarrow} 0 \nonumber 
\end{eqnarray}
 
where $(b(t),d(t))$ is the unique solution of the o.d.e..

\[\dot{b}(t) = d(t) \]
\[\dot{d(t)} = -d(t) + \frac{\Gamma d(t)}{1 - \Gamma b(t)} (1 - b(t) - d(t))\]

with initial conditions $(b(0)=0,d(0)=a(0))$.

\begin{flushright}
$\blacksquare$
\end{flushright}

\section{Solution of the o.d.e.}
\label{app:ode-solving}
Recall the system of o.d.e.s for the evolution of influence under the uniform distribution is given by, 
\[ \dot{b} = d\]
\[ \dot{d} = -d + \frac{\Gamma d}{1 - \Gamma b} (1 -b - d) \]
Substituting for $d$ in the second equation and simplifying, we get
\[\ddot{b} = \frac{\Gamma \dot{b} -\dot{b} -\Gamma \dot{b}^2}{1-\Gamma b}\]
Note that,
\begin{eqnarray*}
\ddot{b} &=& \frac{\mathrm{d}d}{\mathrm{d}t} = \frac{\mathrm{d}d}{\mathrm{d}b} \times \frac{\mathrm{d}b}{\mathrm{d}t}\\
         &=& d \frac{\mathrm{d}d}{\mathrm{d}b}\\
\end{eqnarray*}
Hence,
\[ d\frac{\mathrm{d}d}{\mathrm{d}b} = \frac{\Gamma \dot{b} -\dot{b} -\Gamma \dot{b}^2}{1-\Gamma b}\]
By separating the variables,
\[ \frac{\mathrm{d}d}{\Gamma -1-\Gamma d} = \frac{\mathrm{d}b}{1-\Gamma b} \]
Integrating on both sides, and after taking anti-logarithm
\[ \Gamma - 1 - \Gamma d = c_1 (1-\Gamma b) \]
Differentiating on both sides yields $\dot{d} = c_1 d$ and hence $d(t) = c_2 e^{c_1 t}$. Substituting in the above equation for $d(t)$ we get 
\[b(t) = \frac{c_2}{c_1} e^{c_1 t} + \frac{1+ c_1 - \Gamma}{\Gamma c_1} \]
Solving for constants using the initial conditions $b(0)=0$, $d(0)=d_0$ results in $ c1= -(1 + \Gamma d_0 -\Gamma) =: -r$ and $c_2=d_0$.
Hence we have,
\[b(t) = \frac{d_0}{r} - \frac{d_0}{r} e^{-rt} \]
\[d(t) = d_0 e^{-rt} \]
\begin{flushright}
$\blacksquare$
\end{flushright}

\section{Convergence of $h_{\gamma}^{(N)}(k)$ to $b_\infty$}
\label{app:hilt-convergence}
The solution of the o.d.e. suggests that $\lim_{t \to \infty} b(t) = d_0/r$. This is consistent with the fact that $\lim_{t \to \infty} \frac{h_{\gamma}^{(N)}(k)}{N} \rightarrow \frac{d_0}{1 - (1-d_0)\Gamma} = b_\infty$ as we now proceed to show.
\[ h_{\gamma}^{(N)}(k) = k \bigg[ 1 + (N-k)\gamma \bigg[ 1 + (N-k-1)\gamma \bigg[ 1 + \cdots \]
\[ h_{\gamma}^{(N)}(k+1) = (k+1) \bigg[ 1 + (N-k-1)\gamma \bigg[ 1 + \cdots \]
Thus we can write,
\[ h_{\gamma}^{(N)}(k) = k \bigg[ 1 + \gamma (N-k) \frac{h_{\gamma}^{(N)}(k+1)}{k+1} \bigg] \]
Now substituting for $k=d_0 N$ and noting that $\Gamma = \gamma N$ we have
\[  \frac{h_{\gamma}^{(N)}(N d_0)}{N d_0} = 1 + \Gamma (1-d_0) \frac{h_{\gamma}^{(N)}(N(d_0 + \frac{1}{N}))}{N(d_0 + \frac{1}{N})} \]
Taking $N \rightarrow \infty$ and noting that $h_{\gamma}^{(N)}(k)$ is a continuous function, we have
\[  \frac{1}{d_0} \frac{h_{\gamma}^{(N)}(N d_0)}{N} = 1 + \frac{\Gamma (1-d_0)}{d_0} \frac{h_{\gamma}^{(N)}(N d_0)}{N} \]
Take limits on both sides, and solving for the unknown, 
\[\frac{h_{\gamma}^{(N)}(N d_0)}{N} \rightarrow \frac{d_0}{1 - (1-d_0)\Gamma} = b_\infty\]
\begin{flushright}
$\blacksquare$
\end{flushright}
\end{appendices}








\end{document}